\documentclass{entcs}
\usepackage{entcsmacro}
\usepackage{graphicx}
\usepackage{wrapfig}
\usepackage[all]{xy}
\usepackage{latexsym}
\usepackage{amssymb}
\usepackage{amsmath}
\usepackage{etex}
\usepackage{comment}
\usepackage{bussproofs}
\usepackage{braketmod}
\usepackage{algorithm}
\usepackage[noend]{algpseudocode}

\newcommand{\Set}{\mathbf{Set}}

\newcommand{\BRel}{\mathbf{BRel}}

\newcommand{\SRel}{\mathbf{SRel}}
\newcommand{\Meas}{\mathbf{Meas}}

\newcommand{\wCPO}{\omega\mathbf{CPO}}

\newcommand{\opposite}[1]{{{#1}}^{\mathrm{o\!p}}}
\newcommand{\inverse}[1]{{#1}^{\hspace{-0.15em} - \hspace{-0.1em}1}\hspace{-0.1em}}

\newcommand{\supp}{\mathop{\mathrm{supp}}}

\newcommand{\lrangle}[1]{\langle #1\rangle}
\newcommand{\interpret}[1]{{[\![ {#1} ]\!]}}

\def\lastname{Sato}
\begin{document}
\begin{frontmatter}
  \title{Approximate Relational Hoare Logic for Continuous Random Samplings}
	\author{Tetsuya Sato\thanksref{myemail}}
  \address{Research Institute for Mathematical Sciences, Kyoto University, Kyoto, 606-8502, Japan} 
	\thanks[myemail]{Email:\href{mailto:satoutet@kurims.kyoto-u.ac.jp}
		{\texttt{\normalshape satoutet@kurims.kyoto-u.ac.jp}}} 
\begin{abstract}
Approximate relational Hoare logic (apRHL) is a logic for formal verification of the differential privacy of databases written in the programming language pWHILE.
Strictly speaking, however, this logic deals only with discrete random samplings.
In this paper, we define the graded relational lifting of the subprobabilistic variant of Giry monad, which described differential privacy.
We extend the logic apRHL with this graded lifting to deal with continuous random samplings.
We give a generic method to give proof rules of apRHL for continuous random samplings.
\end{abstract}
\end{frontmatter}
\section{Introduction}
Differential privacy is a \emph{definition} of privacy of \emph{randomized} databases proposed by Dwork, McSherry, Nissim and Smith \cite{DworkMcSherryNissimSmith2006}.
A randomized database satisfies $\varepsilon$-differential privacy (written   $\varepsilon$-differentially private) if for any two adjacent data, the difference of their output probability distributions is bounded by the privacy strength $\varepsilon$.
Differential privacy guarantees high secrecy against database attacks regardless of the attackers' background knowledge, and it has the composition laws, with which we can calculate the privacy strength of a composite database from the privacy strengths of its components.

\emph{Approximate relational Hoare logic} (apRHL) \cite{Barthe:2012:PRR:2103656.2103670,olmedo2014approximate} is a probabilistic variant of the \emph{relational Hoare logic} \cite{Benton2004export:67345} for formal verification of the differential privacy of databases written in the programming language pWHILE.
In the logic apRHL, a parametric relational lifting, which relate probability distributions, play a central role to describe differential privacy in the framework of verification.
This parametric lifting is an extension of the relational lifting \cite[Section 3]{hughesjacobs} that captures probabilistic bisimilarity of Markov chains \cite{LarsenSkou1991} (see also \cite[lemma 4]{deVink1999271}).
The concept of differential privacy is described in the category of binary relation and mappings between them, and verified by the logic apRHL.

Strictly speaking, however, apRHL deals only with random samplings of \emph{discrete} distributions, while the algorithms in many actual studies for differential privacy are modelled with \emph{continuous} distributions, such as, the Laplacian distributions over real line.
Therefore apRHL is desired to be extended to deal with random continuous samplings.
\subsection{Contributions}
Main contributions of this paper are the following two points:
\begin{itemize}
\item We define the graded relational lifting of sub-Giry monad describing differential privacy for continuous random samplings.
\item We extend the logic apRHL \cite{Barthe:2012:PRR:2103656.2103670,olmedo2014approximate} for continuous random samplings (we name \emph{continuous apRHL}) .
\end{itemize}
This graded relational lifting is developed
without witness distributions of probabilistic coupling, and hence is constructed in a different way from the coupling-based parametric lifting of relations given in the studies of apRHL \cite{2016arXiv160105047B,Barthe:2012:PRR:2103656.2103670,olmedo2014approximate}.

In the continuous apRHL, we mainly extend the proof rules for relation compositions and the frame rule.
We also develop a generic method to construct proof rules for random samplings.
By importing the new rules added to apRHL+ in \cite{2016arXiv160105047B},
we give a formal proof of the differential privacy of the \emph{above-threshold algorithm} for real-valued queries \cite[Section 3.6]{DworkRothTCS-042}.
\subsection{Preliminaries}
We denote by $\Meas$ the category of measurable spaces and measurable functions between them and denote by $\Set$ the category of all sets and functions.
The category $\Meas$ is complete and cocomplete, and the forgetful functor
$U \colon \Meas \to \Set$ preserves products and coproducts.
We also denote by $\wCPO_\bot$ of the cateory of $\omega$-complete partial orders with the least element and continuous functions.
\paragraph{A Category of Relations between Measurable Spaces}
We introduce the category $\BRel(\Meas)$ of binary relations between measurable spaces as follows:
\begin{itemize}
\item An object is a triple $(X,Y,\Phi)$ consisting of measurable spaces $X$ and $Y$ and a relation $R$ between $X$ and $Y$ (i.e. $R \subseteq UX \times UY$).
We remark that $R$ does not need to be a measurable subset of the product space $X \times Y$.
\item An arrow $(f,g) \colon (X,Y,\Phi) \to (X',Y',\Phi')$ is a pair of measurable functions $f \colon X \to X'$ and $g \colon Y \to Y'$ such that $(Uf \times Ug)(\Phi) \subseteq \Phi'$.
\end{itemize}
When we write an object $(X,Y,\Phi)$ in $\BRel(\Meas)$, we omit writing the underlying spaces $X$ and $Y$ if they are obvious from the context.
We write $p$ for the forgetful functor $p \colon \BRel(\Meas) \to \Meas \times \Meas$ which extracting underlying spaces: $(X,Y,\Phi) \mapsto (X,Y)$.
We call an endofunctor $F$ on $\BRel(\Meas)$ a \emph{relational lifting} of an endofunctor $E$ on $\Meas$ if $(E \times E)p = pF$.
\paragraph{The Sub-Giry Monad}\label{sec:subgiry}
The Giry monad on $\Meas$ is
introduced in \cite{Giry1982} to give a categorical approach to probability theory; each arrow $X \to Y$ in the Kleisli category of the Giry monad bijectively corresponds to a probabilistic transition from $X$ to $Y$, and the Chapman-Kolmogorov equation corresponds to the associativity law of the Giry monad.

We recall the sub-probabilistic variant of the Giry monad, which we call the \emph{sub-Giry monad} (see also \cite[Section 4]{Panangaden1999171}):
\begin{itemize}
	\item For any measurable space $(X,\Sigma_X)$,
		the measurable space $(\mathcal{G}X,\Sigma_{\mathcal{G}X})$ is defined as follows:
		the underlying set $\mathcal{G}X$ is the set of subprobability measures over $X$,
		and the $\sigma$-algebra $\Sigma_{\mathcal{G}X}$ is the coarsest one
		that makes the evaluation function $\mathrm{ev}_A \colon \mathcal{G}X \to [0,1]$ (mapping $\nu$ to $\nu(A)$) measurable for each $A \in \Sigma_X$.
	\item For each $f \colon X \to Y$ in $\Meas$, $\mathcal{G}f \colon \mathcal{G}X \to \mathcal{G}Y$ is defined by
		$(\mathcal{G}f)(\nu) = \nu(\inverse{f}(-))$.
	\item The unit $\eta$ is defined by $\eta_X (x) = \delta_{x}$,
		where $\delta_{x}$ is the \emph{Dirac measure} centred on $x$. 
	\item The multiplication $\mu$ is defined by $\mu_X(\Xi)(A) = \int_{\mathcal{G}X} \mathrm{ev}_A~d(\Xi)$.
		The Kleisli lifting
		of $f \colon X \to \mathcal{G}Y$
		is given by $f^\sharp(\nu)(A) = \int_{X} f({-})(A)~d\nu$ ($\nu \in \mathcal{G}X$).
\end{itemize}
The monad $\mathcal{G}$ is commutative strong with respect to the cartesian product in $\Meas$.
The strength $\mathrm{st}_{-,=} \colon ({-})\times \mathcal{G}({=}) \Rightarrow \mathcal{G}({-} \times {=})$ is given by the product measure $\mathrm{st}_{X,Y}(x,\nu) = \delta_x \otimes \nu$.
The commutativity of $\mathcal{G}$ is given from the Fubini theorem.
The double strength $\mathrm{dst}_{-,=} \colon \mathcal{G}({-})\times \mathcal{G}({=}) \Rightarrow \mathcal{G}({-} \times {=})$ is given by $\mathrm{dst}_{X,Y}(\nu_1,\nu_2) = \nu_1 \otimes \nu_2$.

The Kleisli category $\Meas_\mathcal{G}$ is often called the category $\mathbf{SRel}$ of \emph{stochastic relations} \cite[Section 3]{Panangaden1999171}.
The category $\mathbf{SRel}$ is $\wCPO_\bot$-enriched (with respect to the cartesian monoidal structure) with the following pointwise order:
\[
f \sqsubseteq g \iff \forall{x \in X, B \in \Sigma_Y}. f(x)(B) \leq g(x)(B)
\quad (f,g \colon X \to Y \text{ in } \SRel).
\]
The \emph{least upper bound} $\sup_{n\in\mathbb{N}} f_n$ of any $\omega$-chain $f_0 \sqsubseteq f_1 \sqsubseteq \cdots \sqsubseteq f_n \sqsubseteq \cdots$ is given by $(\sup_{n} f_n)(x)(B) = \sup_{n} (f_n(x)(B))$.
The \emph{least function} of each $\mathbf{SRel}(X,Y)$ (written $\bot_{X,Y}$) is the constant function of the null-measure over $Y$.
The \emph{continuity} of composition is obtained from the following two facts:
\begin{itemize}
\item
From the definition of Lebesgue integral, for any $\omega$-chain $\{\nu_n\}$ of subprobability measures over $X$, 
$\int_X f ~d(\sup_n \nu_n) = \sup_n \int_X f ~d\nu_n$ holds.
\item From the monotone convergence theorem, we have $\int_X \sup_n f_n ~d\nu = \sup_n \int_X f_n ~d\nu$.
\end{itemize}
This enrichment is equivalent to the partially additive structure on $\SRel$
\cite[Section 5]{Panangaden1999171}:
For any $\omega$-chain
$\{f_n\}_{n\in\mathbb{N}}$ of $f_n \colon X \to Y$ in $\SRel$,
we have the summable sequence $\{g_n\}_{n}$ where $g_0 = f_0$ and $g_{n+1} = f_{n+1} - f_n$.
Conversely, for any summable sequence $\{g_n\}_{n\in\mathbb{N}}$, the functions $f_n = \sum_{k=0}^{n} g_n$ form an $\omega$-chain.
\paragraph{Differential privacy}
Throughout this paper, we define the approximate differential privacy as follows:
\begin{definition}[{\cite[Definition 2.4]{DworkRothTCS-042}}, Modified]
A measurable function $c \colon \mathbb{R}^m \to \mathcal{G}(\mathbb{R}^n)$
is $(\varepsilon,\delta)$-differentially private if $c(x)(A) \leq \exp(\varepsilon)c(y)(A) + \delta$ holds for any $||x - y||_1 \leq 1$ and $A \in \Sigma_{\mathbb{R}^n}$, where $|| \cdot ||_1$ is $1$-norm of the Euclidean space $\mathbb{R}^m$.
\end{definition}
What we modify from the original definition \cite[Definition 2.4]{DworkRothTCS-042} is the domain and codomain of $c$; we replace the domain from $\mathbb{N}$ to $\mathbb{R}$, and replace the codomain from a discrete probability space to $\mathcal{G}(\mathbb{R}^n)$.
We apply this definition to the interpretation of pWHILE programs.
The input and output spaces can be other spaces:
in section \ref{sec:example}  we consider the \emph{above-threshold algorithm} $\mathtt{Above}$ whose output space is $\mathbb{Z}$.
The above modification is essential in describing and verifying the differential privacy of this algorithm because it takes a sample from Laplace distribution over \emph{real line}.
\section{A Graded Monad for Differential Privacy}
The composition law of differential privacy plays crucial role to in the compositional verification of the differential privacy of database programs.
Barthe, K\"{o}pf, Olmedo, and Zanella-B{\'e}guelin constructed a \emph{parametric relational lifting} describing differential privacy, and developed a framework for compositional verification of differential privacy \cite{Barthe:2012:PRR:2103656.2103670}.

Following this relational approach, we construct the parametric relational lifting of Giry monad to describe differential privacy for continuous random samplings.
This lifting forms a graded monad on the category $\BRel(\Meas)$ in the sense of  \cite{Katsumata2014PEM}.
The axioms of graded monad correspond to the (sequential) composition law of differential privacy.
\subsection{Graded Monads}
\begin{definition}\cite[Definition 2.2-bis]{Katsumata2014PEM}
Let $\mathbb{C}$ be a category, and $(M,\cdot,1,{\preceq})$ be a \emph{preordered} monoid.
An $M$-graded (or $M$-parametric effect) monad on $\mathbb{C}$ consists of
\begin{itemize}
	\item a collection $\{T_{e}\}_{e \in M}$ of endofunctors on $\mathbb{C}$,
	\item a natural transformation $\eta \colon \mathrm{Id} \Rightarrow T_1$,
	\item a collection $\{\mu^{e_1,e_2}\}_{e_1, e_2 \in M}$ of natural transformations $\mu^{e_1,e_2} \colon T_{e_1}T_{e_2} \Rightarrow T_{e_1 e_2}$,
	\item a collection $\{{\sqsubseteq}^{e_1, e_2} \}_{e_1 \preceq  e_2}$ of natural transformations $\sqsubseteq^{e_1, e_2} \colon T_{e_1} \Rightarrow T_{e_2}$
\end{itemize}
satisfying  
\begin{itemize}
	\item $\mu^{e,1} \circ T_e \eta = \mu^{1,e} \circ \eta_{T_e} = \mathrm{Id}_{T_e}$ for any $e \in M$,
	\item $\mu^{(e_1 e_2),e_3} \circ \mu^{e_1,e_2}{T_{e_3}}  = \mu^{e_1,(e_2, e_3)} \circ T_{e_1} \mu^{e_2,e_3}$ for all $e_1, e_2, e_3 \in M$,
	\item ${\sqsubseteq}^{e, e} = \mathrm{Id}_{T_{e}}$ for any $e$ and ${\sqsubseteq}^{e_2, e_3} \circ {\sqsubseteq}^{e_1, e_2} = {\sqsubseteq}^{e_1, e_3}$ whenever $e_1 \preceq e_2 \preceq e_3$,
	\item $\sqsubseteq^{(e_1 e_2), (e_3 e_4)} \circ \mu^{e_1, e_2} = \mu^{e_3, e_4} \circ (\sqsubseteq^{e_1, e_3} \ast {\sqsubseteq}^{e_2, e_4})$ whenever $e_1 \preceq e_3$ and $e_2 \preceq e_4$.
\end{itemize}
\end{definition}
An $M$-graded monad $(\{T_{e}\}_{e \in M}, \eta, \mu^{e_1,e_2}, {\sqsubseteq}^{e_1, e_2})$ on $\mathbb{C}$ is called an $M$-graded lifting of monad $(T,\eta^T,\mu^T)$ on $\mathbb{D}$ along $U \colon \mathbb{C} \to \mathbb{D}$ if $U{T_{e}} = TU$, $U(\eta) = \eta^T U$, $U(\mu^{e_1,e_2}) = \mu^T U$, and $U({\sqsubseteq}^{e_1, e_2}) = \mathrm{id}_T$.
\subsection{A Graded Relational Lifting of Giry Monad for Differential Privacy}\label{sec:gradedlifting}
Let $M$ be the cartesian product of the monoids
$([1,\infty), \times, 1 )$ and $([0,\infty), + , 0)$ equipped with the product order of numerical orders.
For each $(\gamma,\delta) \in M$, we define the following mapping of $\BRel(\Meas)$-objects by
\[
\mathcal{G}^{(\gamma,\delta)} \Phi
	= \SetBraket{
		(\nu_1,\nu_2) \in \mathcal{G}X \times \mathcal{G}Y
		| \begin{array}{l@{}}
			\forall{A \in \Sigma_X, B \in \Sigma_Y}.\\
			\Phi(A)\subseteq B \implies \nu_1(A) \leq \gamma \nu_2(B) + \delta
			\end{array}
	}.
\]
\begin{proposition}
$\{\mathcal{G}^{(\gamma,\delta)}\}_{(\gamma,\delta) \in M}$ forms an $M$-graded lifting of the monad $(\mathcal{G}\times\mathcal{G},\eta\times\eta,\mu\times\mu)$ along the forgetful functor $p \colon \BRel(\Meas) \to \Meas\times\Meas$.
\end{proposition}
\begin{proof}
Since the functor $p$ is faithful, it suffices to show:
\begin{enumerate}
\item \label{enum:gradedlifting1}
	Each $\mathcal{G}^{(\gamma,\delta)}$ is an endofunctor on $\BRel(\Meas)$.
\item \label{enum:gradedlifting2}
	$(\mathrm{id}_{\mathcal{G}X},\mathrm{id}_{\mathcal{G}Y})$ is an arrow
 	$\mathcal{G}^{(\gamma,\delta)} \Phi \to \mathcal{G}^{(\gamma',\delta')} \Phi$ in $\BRel(\Meas)$
	for all $\gamma,\gamma',\delta,\delta'$ such that $\gamma \leq \gamma'$ and $\delta \leq \delta'$.
\item \label{enum:gradedlifting3}

	$(\eta_X, \eta_Y)$ is an arrow
	$\Phi \to \mathcal{G}^{(1,0)}\Phi$ in $\BRel(\Meas)$.
\item \label{enum:gradedlifting4}
	$(\mu_X, \mu_Y)$ is an arrow
	$\mathcal{G}^{(\gamma,\delta)}\mathcal{G}^{(\gamma',\delta')}\Phi \to \mathcal{G}^{(\gamma \gamma', \delta + \delta)}\Phi$ in $\BRel(\Meas)$
	for all $\gamma,\gamma',\delta,\delta'$.
\end{enumerate}
(\ref{enum:gradedlifting1})
Since the mapping $(f,g) \mapsto (\mathcal{G}f,\mathcal{G}g)$ is obviously functorial, it suffices to check that $(\mathcal{G}f,\mathcal{G}g)$ is an arrow $\mathcal{G}^{(\gamma,\delta)}\Psi \to \mathcal{G}^{(\gamma,\delta)}\Phi$ in $\BRel(\Meas)$ for any arrow $(f,g) \colon \Psi \to \Phi$ in $\BRel(\Meas)$.
This is proved from $\Phi(A) \subseteq B \implies \Psi(\inverse{f}(A)) \subseteq \inverse{g}(B)$ for any $A \in \Sigma_X$ and $B \in \Sigma_Y$.
(\ref{enum:gradedlifting2}) Obvious.
(\ref{enum:gradedlifting3}) Obvious.
(\ref{enum:gradedlifting4}) It suffices to show $(\mu_X \times \mu_Y)(\mathcal{G}^{(\gamma,\delta)}\mathcal{G}^{(\gamma',\delta')}\Phi) \subseteq \mathcal{G}^{(\gamma \gamma', \delta + \delta)}\Phi$ for any $\Phi \subseteq X \times Y$.

First, the following equation holds: 
\begin{align*}
	\mathcal{G}^{(\gamma,\delta)} \Phi
	&=\SetBraket{
		(\nu_1,\nu_2) 
		| \forall{(f,g) \colon \Phi \to {\leq} \text { in } \BRel(\Meas)} .
			\int_X f~d\nu_1 \leq \gamma\!\int_Y g~d\nu_2 + \delta
	},
\end{align*}
where $\leq$ is the numerical order relation on $\mathcal{G}1 \simeq [0,1]$.
We omit the proof of this equation.
It can be shown in the same way as \cite[Theorem 12]{katsumata_et_al:LIPIcs:2015:5532}.

Let $(\Xi_1,\Xi_2) \in \mathcal{G}^{(\gamma,\delta)}\mathcal{G}^{(\gamma',\delta')}\Phi$.
Assume $\Phi(A) \subseteq B$.
We give $(f,g)\colon \mathcal{G}^{(\gamma',\delta')}\Phi\to\leq$ in $\BRel(\Meas)$ by 
$f = \max(\mathrm{ev}_A - \delta', 0)$ and $g = \min(\gamma' \cdot\mathrm{ev}_B, 1)$.
They actually satisfy $f(\nu_1) \leq g(\nu_2)$ for each $(\nu_1,\nu_2) \in \mathcal{G}^{(\gamma',\delta')}\Phi$.
Hence,
\begin{align*}
	\mu_X(\Xi_1)(A) - \delta'
	&\leq \int_{\mathcal{G}X} (\mathrm{ev}_A- \delta')~d \Xi_1
	\leq \int_{\mathcal{G}X} f~d \Xi_1\\
	&\leq  \gamma \int_{\mathcal{G}X} g~d \Xi_2 + \delta
	\leq  \gamma \int_{\mathcal{G}X} \gamma' \mathrm{ev}_B~d \Xi_2 + \delta
	= \gamma \gamma' \mu_Y(\Xi_2)(B) + \delta.
\end{align*}
This implies $\mu_X(\Xi_1)(A) \leq \gamma \gamma' \mu_Y(\Xi_2)(B) + \delta + \delta'$.
\end{proof}
The $M$-graded lifting $\{\mathcal{G}^{(\gamma,\delta)}\}_{(\gamma,\delta) \in M}$
describes only one side of inequalities in the definition of differential privacy.
By symmetrising this, we obtain the following $M$-graded lifting $\{\overline{\mathcal{G}^{(\gamma,\delta)}}\}_{(\gamma,\delta) \in M}$ exactly describing the differential privacy for continuous probabilities:
\[
\overline{\mathcal{G}^{(\gamma,\delta)}} = \mathcal{G}^{(\gamma,\delta)}(-) \cap \opposite{(\mathcal{G}^{(\gamma,\delta)}\opposite{(-)})}.
\]
\begin{theorem}\label{thm:privacy_gradedmonad}
A measurable function $c \colon \mathbb{R}^m \to \mathcal{G}(\mathbb{R}^n)$ is $(\varepsilon,\delta)$-differentially private
\emph{if and only if}
$(c,c)$ is an arrow $\SetBraket{(x,y)| {||x - y||_1} \leq 1} \to \overline{\mathcal{G}^{(\exp(\varepsilon),\delta)}}\mathrm{Eq}_{\mathbb{R}^n}
$ in $\BRel(\Meas)$.
\end{theorem}

In the original works \cite{Barthe:2012:PRR:2103656.2103670,BartheOlmedo2013} of apRHL,
the following relational lifting $(-)^{\sharp(\gamma,\delta)}$ is introduced to describe differential privacy.
This lifting relates two distributions if there are intermediate distributions $d_1$ and $d_R$, called \emph{witnesses}, whose skew distance, defined by
$
\Delta^X_\gamma(d_L,d_R) = \sup_{C \subseteq X}\left\{\left| d_L(C) - \gamma d_R(C) \right|, \left| d_R(C) - \gamma d_L(C) \right|\right\}
$, is less than or equal to $\delta$.
\begin{definition}(\cite[Definition 4]{BartheOlmedo2013}, \cite[Definition 4.3]{olmedo2014approximate} and \cite[Definition 8]{2016arXiv160105047B})\label{def:lifting:witness}
We denote by $\mathcal{D}$ the subdistribution monad over $\Set$.
Let $\Psi$ be a relation between sets $X$ and $Y$, and
$d_1 \in \mathcal{D}X$ and $d_2 \in \mathcal{D}Y$ be two subdistributions.
We define the relation $\mathbin{\Psi^{\sharp(\gamma,\delta)}} \subseteq \mathcal{D}X \times \mathcal{D}Y$ as follows: $(d_1, d_2) \in \mathbin{\Psi^{\sharp(\gamma,\delta)}}$ if and only if
there are two subdistributions $d_L, d_R \in \mathcal{D}(X \times Y)$, called \emph{witnesses}, such that
\[
\mathcal{D}\pi_1(d_L) = d_1,~
\mathcal{D}\pi_2(d_R) = d_2,~
\mathrm{supp}(d_L) \subseteq \Psi,~
\mathrm{supp}(d_R) \subseteq \Psi,~
\Delta^{X \times Y}_{\gamma}(d_L,d_R) \leq \delta.
\]
\end{definition}
\begin{proposition}\label{prop:lifting:compare}
For any countable discrete spaces $X$ and $Y$, and relation $\Psi \subseteq X \times Y$, 
we have $\Psi^{\sharp(\gamma,\delta)} \subseteq \overline{\mathcal{G}^{(\gamma,\delta)}} \Psi$.
\end{proposition}
\begin{proof}
Suppose $(d_1, d_2) \in \mathbin{\Psi^{\sharp(\gamma,\delta)}}$ with witnesses $d_L$ and $d_R$.
For any $A \subseteq X$, since $\supp(d_L) \subseteq \Psi$ and $(A \times Y) \cap \Psi \subseteq X \times \Psi(A)$, we obtain:
\begin{align*}
d_1(A)
&=
\mathcal{D}\pi_1(d_L)(A)
=d_L(A \times Y)
=d_L((A \times Y) \cap \Psi)
\leq d_L(X \times \Psi(A))\\
&\leq \gamma d_R (X \times \Psi(A)) +\delta
= \gamma \mathcal{D}\pi_2(d_R)(\Psi(A)) +\delta
= \gamma d_2(\Psi(A)) + \delta.
\end{align*}
This implies $(d_1,d_2) \in \mathcal{G}^{(\gamma,\delta)}\Psi$.
Since the construction of $(-)^{\sharp(\gamma,\delta)}$ is symmetric,
we conclude $(d_1,d_2) \in \overline{\mathcal{G}^{(\gamma,\delta)}}\Psi$.
\end{proof}
We remark $\mathcal{G}X = \mathcal{D}X$ for countable discrete space $X$.
When $X$ is not countable, we have the above results by embedding each $d \in \mathcal{D}X$ in the set $\mathcal{D}X'$ of subprobability distributions over the countable \emph{subspace} $X' = X \cap \supp(d)$. 
\begin{corollary}
We have $\mathrm{Eq}_X^{\sharp(\gamma,\delta)} = \overline{\mathcal{G}^{(\gamma,\delta)}}\mathrm{Eq}_X$ for each countable discrete space $X$.
\end{corollary}
\begin{proof}
($\subseteq$)
This inclusion is given from Proposition \ref{prop:lifting:compare}.
($\supseteq$)
Suppose $(d_1,d_2) \in \overline{\mathcal{G}^{(\gamma,\delta)}}\mathrm{Eq}_X$.
This is equivalent to $\Delta^X_\gamma (d_1,d_2)\leq \delta$.
Hence $(d_1, d_2) \in \mathrm{Eq}_X^{\sharp(\gamma,\delta)}$ is proved by the witnesses given by $d_L = \sum_{x \in X} d_1(x) \cdot \delta_{(x,x)}$ and $d_R = \sum_{x \in X} d_2(x) \cdot \delta_{(x,x)}$.
\end{proof}
\section{The Continuous apRHL}
We introduce a variant of the approximate probabilistic relational Hoare logic (apRHL) to deal with continuous random samplings.
We name it the \emph{continuous apRHL}.
\subsection{The Language pWHILE}
We recall and reformulate categorically the language pWHILE \cite{Barthe:2012:PRR:2103656.2103670}.
In this paper, we mainly refer to the categorical semantics of a probabilistic language given in \cite[Section 2]{Brown2009193}.
The language pWHILE is constructed in the standard way, hence we sometimes omit the details of its construction.
\subsubsection{Syntax}
We introduce the syntax of pWHILE by the following BNF:
\begin{align*}
	\tau &
		::= \mathtt{bool}
		\mid \mathtt{int}
		\mid \mathtt{real}
		\mid \ldots
		\\
	e &
		::= x
		\mid p (e_1, \ldots, e_m)
		\\
	\nu &
		::=
		d (e_1, \ldots, e_m)
		\\
	i &
		::= x \leftarrow e
		\mid x \xleftarrow{\$} \nu
		\mid \mathtt{if}~e~\mathtt{then}~c_1~\mathtt{else}~c_2
		\mid \mathtt{while}~e~\mathtt{do}~c
		\\
	c &
		::= \mathtt{skip}
		\mid \mathtt{null}
		\mid \mathcal{I} ; \mathcal{C}
\end{align*}
Here, $\tau$ is a \emph{value type}; $x$ is a \emph{variable}; $p$ is an \emph{operation}; $d$ is a \emph{probabilistic operation}; $e$ is an \emph{expression}; $\nu$ is a \emph{probabilistic expression}; $i$ is an \emph{imperative}; $c$ is a \emph{command} (or program).
We remark constants are $0$-ary operations.

We introduce the following syntax sugars for simplicity:
\begin{align*}
	\mathtt{if}~b~\mathtt{then}~c
		&
		= \mathtt{if}~b~\mathtt{then}~c~\mathtt{else}~\mathtt{skip}
		\\
	[\mathtt{while}~b~\mathtt{do}~c]_n
		&
		=
		\begin{cases} 
			\mathtt{if}~b~\mathtt{then}~\mathtt{null}~\mathtt{else}~\mathtt{skip}, & \text{ if } n=0
		\\
			\mathtt{if}~b~\mathtt{then}~c;[\mathtt{while}~b~\mathtt{do}~c]_{k}, & \text{ if } n=k+1
		\end{cases}
\end{align*}
\subsubsection{Typing Rules}
We introduce a typing rule on the language pWHILE.
A typing context is a finite set 
$\Gamma = \{x_1\colon\tau_1,x_2\colon\tau_2,\ldots,x_n\colon\tau_n \}$
of pairs of a variable and a value type such that each variable occurs only once in the context.

We give typing rules of pWHILE as follows:
\[
	\AxiomC{$\Gamma \vdash^t e_1\colon\tau_1 ~\ldots~ \Gamma \vdash^t e_n\colon\tau_n$
	\quad
	$p \colon  (\tau_1,\ldots, \tau_n)  \to \tau$}
	\UnaryInfC{$\Gamma \vdash^t p(e_1,\ldots,e_n) \colon \tau$}
	\DisplayProof
\quad
	\AxiomC{$\Gamma,x\colon\tau \vdash^t e \colon\tau$}
	\UnaryInfC{$\Gamma, x\colon\tau  \vdash x \leftarrow e$}
	\DisplayProof
\quad
	\AxiomC{\hspace{1em}}
	\UnaryInfC{$\Gamma \vdash \mathtt{skip}$}
	\DisplayProof
\]
\[
	\AxiomC{
	$x\colon\tau\in \Gamma$\quad
	$\Gamma \vdash^t e_1\colon\tau_1 ~\ldots~ \Gamma \vdash^t e_n\colon\tau_n$\quad
	$d \colon  (\tau_1,\ldots, \tau_n)  \to \tau$}
	\UnaryInfC{$\Gamma \vdash x \xleftarrow{\$}  d(e_1,\ldots,e_n) \colon \tau$}
	\DisplayProof
\quad
	\AxiomC{\hspace{1em}}
	\UnaryInfC{$\Gamma \vdash \mathtt{null}$}
	\DisplayProof
\]
\[
	\AxiomC{$\Gamma \vdash i$\quad$\Gamma \vdash c$}
	\UnaryInfC{$\Gamma \vdash i ; c$}
	\DisplayProof
\quad
	\AxiomC{$\Gamma \vdash^t b\colon\mathtt{bool}$\quad$\Gamma \vdash c_1$\quad$\Gamma \vdash c_2$}
	\UnaryInfC{$\Gamma \vdash \mathtt{if}~b~\mathtt{then}~c_1~\mathtt{else}~c_2$}
	\DisplayProof
\quad
	\AxiomC{$\Gamma \vdash^t b\colon\mathtt{bool}$\quad$\Gamma \vdash c$}
	\UnaryInfC{$\Gamma \vdash \mathtt{while}~b~\mathtt{do}~c$}
	\DisplayProof
\]
Here, the type $(\tau_1,\ldots, \tau_n) \to \tau$ of each operation $p$ and each probabilistic operation $d$ are assumed to be given in advance.

We easily define inductively the set of free variables of commands, expressions, and  probabilistic expressions (denoted by $FV(c)$, $FV(e)$, and $FV(\nu)$).
\subsubsection{Denotational Semantics}
We introduce a denotational semantics of pWHILE in $\Meas$.
We give the interpretations $\interpret{\tau}$ of the value types $\tau$:
\begin{itemize}
\item $\interpret{\mathtt{bool}} = \mathbb{B} = 1 + 1 = \{\mathtt{true},\mathtt{false}\}$ (discrete space)
\item $\interpret{\mathtt{int}} = \mathbb{Z}$ (discrete space)
\item $\interpret{\mathtt{real}} = \mathbb{R}$ (Lebesgue measurable space)
\end{itemize}
We interpret a typing context $\Gamma = \{x_1\colon\tau_1,x_2\colon\tau_2,\ldots,x_n\colon\tau_n \}$ as the product space $\interpret{\tau_1} \times \interpret{\tau_2} \times \cdots \times \interpret{\tau_n}$.
We interpret each operation $p \colon (\tau_1, \ldots \tau_m) \to \tau$ as
a measurable function $\interpret{p}\colon \interpret{\tau_1}\times\cdots\times\interpret{\tau_m} \to \interpret{\tau}$,
and each probabilistic operation $d \colon (\tau_1, \ldots \tau_m) \to \tau$ as
$\interpret{d}\colon \interpret{\tau_1}\times\cdots\times\interpret{\tau_m} \to \mathcal{G}\interpret{\tau}$.
Typed terms$\Gamma \vdash^t e \colon \tau$ and commands $\Gamma \vdash c$ are interpreted to measurable functions of the forms $\interpret{\Gamma} \to \interpret{\tau}$ and $\interpret{\Gamma} \to \mathcal{G}\interpret{\Gamma}$ respectively.

The interpretation of expressions are defined inductively by:
\begin{align*}
\interpret{\Gamma \vdash^t x\colon\tau} = \pi_{x\colon\tau}
\quad \interpret{\Gamma \vdash^t p(e_1,\ldots,\e_m)} &= \interpret{p}(\interpret{\Gamma \vdash^t e_1},\ldots \interpret{\Gamma \vdash^t e_m})
\end{align*}
The interpretation of commands are defined inductively by:
\[
\interpret{\Gamma\vdash\mathtt{skip}}
	= \eta_{\interpret{\Gamma}}
	\quad
\interpret{\Gamma\vdash\mathtt{null}}
	=\bot_{\interpret{\Gamma},\interpret{\Gamma}}
	\quad
\interpret{\Gamma\vdash i ; c}
	= {(\interpret{\Gamma\vdash c})}^\sharp \circ \interpret{\Gamma\vdash i}
\]
\vspace{-2.4\baselineskip}
\begin{align*}
\lefteqn{\interpret{{\Gamma} \vdash x \xleftarrow{\$} d(e_1,\ldots,\e_m)}}\\
	&= \mathcal{G}(\rho_{(x\colon\tau, \Gamma)})
	\circ\mathrm{st}_{\interpret{\tau},\interpret{\Gamma}}
	\circ \lrangle{\interpret{d}(\interpret{\Gamma \vdash^t e_1},\ldots \interpret{\Gamma \vdash^t e_m}),\mathrm{id}_{\interpret{\Gamma}}}
\end{align*}
\vspace{-2.4\baselineskip}
\begin{align*}
\interpret{{\Gamma, x\colon\tau}\vdash x \leftarrow e}
	&
	= \eta_{\interpret{\Gamma, x\colon\tau}}
	\circ \rho_{(x\colon\tau, \Gamma)}
	\circ \lrangle{ \interpret{{\Gamma,x\colon\tau} \vdash e} ,\mathrm{id}_{\interpret{\Gamma, x\colon\tau}}}
	\\
\interpret{\Gamma \vdash \mathtt{if}~b~\mathtt{then}~c_1~\mathtt{else}~c_2}
	&
	=\left[\interpret{\Gamma\vdash c_1},\interpret{\Gamma\vdash c_2} \right]
	\circ \cong_{\interpret{\Gamma}}
	\circ\lrangle{\interpret{\Gamma\vdash b} ,\mathrm{id}_{\interpret{\Gamma}}}
	\\
\interpret{\Gamma \vdash \mathtt{while}~b~\mathtt{do}~c}
	&
	= \sup_{n\in\mathbb{N}} \interpret{\Gamma\vdash [\mathtt{while}~e~\mathtt{do}~c]_n}
\vspace{-\baselineskip}
\end{align*}
Here,
\begin{itemize}
\item
	$
	\rho_{(x_k\colon\tau_k, \Gamma)}
	= \lrangle{f_l}_{l \in \{1,2,\ldots,n\}}
	\colon \interpret{\tau_k} \times \interpret{\Gamma} \to \interpret{\Gamma}
	$,
	where $\Gamma = \{x_1\colon\tau_1,x_2\colon\tau_2,\ldots,x_n\colon\tau_n\}$,
	$f_k = \pi_2$, and $f_l = \pi_l \circ \pi_2$ ($l \neq k$).
\item
	${\cong_X} \colon 2 \times X \to X+X$ is
	the inverse of $[\lrangle{\iota_1 \circ !_X,~ id}, \lrangle{\iota_2 \circ !_X,~ id}] \colon X + X \to 2 \times X$, which is obtained from the distributivity of the category $\Meas$.
\end{itemize}
We remark that, from the commutativity of the monad $\mathcal{G}$, if $\Gamma \vdash x \colon \tau$ and $x \notin FV(c)$ then $\interpret{\Gamma \vdash c} \cong \mathrm{dst}_{\interpret{\Gamma'},\interpret{\tau}}(\interpret{\Gamma' \vdash c}\times \eta_\interpret{\tau})$ where $\Gamma' = \Gamma \setminus \{x \colon \tau\}$.
\subsection{Judgements of apRHL}
A judgement of apRHL is
\[
	c_1 \sim_{\gamma, \delta} c_2 \colon \Psi \Rightarrow \Phi,
\]
where $c_1$ and $c_1$ are commands, and $\Psi$ and $\Phi$ are objects in $\BRel(\Meas)$.
We call the relations $\Psi$ and $\Phi$ the \emph{precondition} and \emph{postcondition} of the judgement respectively.
Inspired from the validity of asymmetric apRHL \cite{Barthe:2012:PRR:2103656.2103670}, we introduce the validity of the judgement of apRHL.
\begin{definition}
	Let $\Psi$ and $\Phi$ be relations over the space $\interpret{\Gamma}$.
	A judgement $c_1 \sim_{\gamma, \delta} c_2 \colon \Psi \Rightarrow \Phi$
	is valid (written $\models c_1 \sim_{\gamma, \delta} c_2 \colon \Psi \Rightarrow \Phi$)
	when $(\interpret{\Gamma \vdash c_1}, \interpret{\Gamma \vdash c_2})$ is an arrow  $\Psi \to \overline{\mathcal{G}^{(\gamma, \delta)}} \Phi$ in $\BRel(\Meas)$.
\end{definition}
We often write preconditions and postconditions in the following manner: 
Let $\Gamma = \{x_1\colon\tau_1,x_2\colon\tau_2,\ldots,x_n\colon\tau_n\}$.
Assume $\Gamma \vdash e_1 \colon \tau$ and $\Gamma \vdash e_2 \colon \tau$, and let $R$ be a relation on $\interpret{\tau}$ (e.g. $=$, $\leq$,... ).
We define the relation $e_1 \lrangle{1} R e_2 \lrangle{2}$ on $\interpret{\Gamma}$ by
\[
(e_1 \lrangle{1} R e_2 \lrangle{2})
	=\SetBraket{ (m_1,m_2) \in \interpret{\Gamma}| \interpret{\Gamma\vdash e_1}(m_1) R \interpret{\Gamma\vdash e_2}(m_2)}.
\]
\subsection{Proof Rules}
We mainly refer the proof rules of apRHL from \cite{Barthe:2012:PRR:2103656.2103670,olmedo2014approximate}, but we modify the [comp] and [frame] rules to verify differential privacy for continuous random samplings.
\[
	\AxiomC{$
		\begin{array}{l@{}}
			{x_1 \colon \tau_1}, {x_2 \colon \tau_2} \in \Gamma
				\quad
			\Gamma \vdash^t e_1 \colon \tau_1
				\quad
			\Gamma \vdash^t e_2 \colon \tau_2
				\\
			(\rho_{(x_1 \colon \tau_1,\Gamma)} \circ \lrangle{\interpret{e_1}, \mathrm{id}},  \rho_{(x_2 \colon \tau_2,\Gamma)} \circ \lrangle{\interpret{e_2}, \mathrm{id}}) \colon \Psi \to \Phi 
		\end{array}
	$}
	\RightLabel{{[assn]}}
	\UnaryInfC{$\models x_1 \leftarrow e_1 \sim_{(1,0)}  x_2 \leftarrow e_2 \colon \Psi \Rightarrow \Phi$}
	\DisplayProof
\]
\[
	\AxiomC{$
		\begin{array}{l@{}}
			\Gamma \vdash^t e^1_1 \colon \tau ~\ldots~ \Gamma \vdash^t e^1_m \colon \tau
				\quad
			\Gamma \vdash^t e^2_1 \colon \tau ~\ldots~ \Gamma \vdash^t e^2_m \colon \tau
				\quad
			x_1 \colon \tau, x_2 \colon \tau \in \Gamma 
				\\
			d \colon (\tau_1,\ldots,\tau_m) \to \tau
				\quad
			(\interpret{d},\interpret{d})\colon \Psi \to \overline{\mathcal{G}^{(\gamma,\delta)}}(\mathrm{Eq}_{\interpret{\tau}})
			 \text{ in } \BRel(\Meas)
		\end{array}
	$}
	\RightLabel{{[rand]}}
	\UnaryInfC{$\models x_1 \xleftarrow{\$} d(e^1_1,\ldots,e^1_m) \sim_{(\gamma,\delta)}  x_2 \xleftarrow{\$} d(e^2_1,\ldots,e^2_m) \colon \Psi' \Rightarrow (x_1\lrangle{1} = x_2\lrangle{1})$}
	\DisplayProof
\]
where $\Psi' = \SetBraket{((g,a),(h,b))| (a,b)\in \Psi, g,h \in \Gamma'}$ ($\Gamma = \{x_1\colon\tau_1,\ldots,x_k\colon\tau_k\} \cup \Gamma'$).
\[
	\AxiomC{$
		\begin{array}{l@{}}
		\models c_1 \sim_{(\gamma,\delta)} c_2 \colon \Psi \Rightarrow \Phi'\\
		\models c_1' \sim_{(\gamma',\delta')} c_2' \colon \Phi' \Rightarrow \Phi
		\end{array}
		$}
	\RightLabel{{[seq]}}
	\UnaryInfC{$\models c_1;c_1'\sim_{(\gamma\gamma',\delta+\delta')} c_2;c_2' \colon \Psi \Rightarrow\Phi$}
	\DisplayProof
\quad
	\AxiomC{}
	\RightLabel{{[skip]}}
	\UnaryInfC{$\models \mathtt{skip} \sim_{(1,0)} \mathtt{skip} \colon \Phi \Rightarrow \Phi$}
	\DisplayProof
\]
\[
	\AxiomC{$
		\begin{array}{l@{}}
			\Gamma \vdash^t b \colon \mathtt{bool}
			\quad \Gamma \vdash^t b \colon \mathtt{bool}
			\quad \Psi \Rightarrow b\lrangle{1}  = b' \lrangle{2}
			\\
			\models c_1 \sim_{(\gamma,\delta)} c_1' \colon \Psi\wedge b\lrangle{1} \Rightarrow \Phi \quad
			\models c_2 \sim_{(\gamma,\delta)} c_2' \colon \Psi\wedge \neg b\lrangle{1} \Rightarrow \Phi
		\end{array}
	$}
	\RightLabel{{[cond]}}
	\UnaryInfC{$\models \mathtt{if}~b~\mathtt{then}~c_1~\mathtt{else}~c_2 \sim_{(\gamma,\delta)} \mathtt{if}~b'~\mathtt{then}~c_1'~\mathtt{else}~c_2' \colon \Psi\Rightarrow\Phi$}
	\DisplayProof
\]
\[
	\AxiomC{$
		\begin{array}{l@{}}
			\Gamma \vdash^t e \colon \mathtt{int}
			\quad \gamma = \prod_{k=0}^{n-1} \gamma_k
			\quad \delta = \sum_{k=0}^{n-1} \delta_k\\
			\Theta \Rightarrow b_1\lrangle{1}  = b_2\lrangle{2}
			\quad \Theta \wedge e \lrangle{1} \geq n \Rightarrow \neg b_1\lrangle{1}\\
			\forall{k\colon\mathtt{int}}. \models c_1 \sim_{(\gamma_k,\delta_k)} c_2 \colon \Theta \wedge e\lrangle{1} = k \wedge e \lrangle{1} \leq n \implies \Theta \wedge e\lrangle{1} > k\\
		\end{array}
	$}
	\RightLabel{{[while]}}
	\UnaryInfC{$
		\begin{array}{l@{}}
			\models \mathtt{while}~b~\mathtt{do}~c_1 \sim_{(\gamma, \delta)} \mathtt{while}~b'~\mathtt{do}~c_2 \colon~\Theta\wedge b_1\langle 1 \rangle \wedge e\langle 1 \rangle \geq 0 \Rightarrow \Theta \wedge \neg b_1\langle 1 \rangle
		\end{array}
	$}
	\DisplayProof
\]
\[
	\AxiomC{$
		\begin{array}{l@{}}
			\models c_1 \sim_{(\gamma,\delta)} c_2 \colon \Psi\wedge\Theta \Rightarrow \Phi
			\quad
			\models c_1 \sim_{(\gamma,\delta)} c_2 \colon \Psi\wedge\neg\Theta \Rightarrow \Phi
		\end{array}
	$}
	\RightLabel{{[case]}}
	\UnaryInfC{$
		\models c_1 \sim_{(\gamma,\delta)} c_2 \colon \Psi \Rightarrow \Phi
	$}
	\DisplayProof
\]
\[
	\AxiomC{$
			\models c_1 \sim_{(\gamma,\delta)} c_2 \colon \Psi \Rightarrow \Phi~~
\Psi' \Rightarrow \Psi
			~~
			\Phi \Rightarrow \Phi'
	$}
	\RightLabel{{[weak]}}
	\UnaryInfC{$
		\models c_1 \sim_{(\gamma,\delta)} c_2 \colon \Psi' \Rightarrow \Phi'
	$}
	\DisplayProof
\quad
	\AxiomC{$
		\models c_1 \sim_{(\gamma,\delta)} c_2 \colon \Psi \Rightarrow \Phi
	$}
	\RightLabel{{[op]}}
	\UnaryInfC{$
		\models c_2 \sim_{(\gamma,\delta)} c_1 \colon \opposite{\Psi} \Rightarrow \opposite{\Phi}
	$}
	\DisplayProof
\]
The relational lifting $\overline{\mathcal{G}^{(\gamma, \delta)}}$ does not preserve every relation composition.
However, it preserve the composition of relations if the relations are \emph{measurable}, that is, the images and inverse images along them of mesurable sets are also measurable
(see also \cite[Section 3.3]{katsumata_et_al:LIPIcs:2015:5532}).
Generally speaking, it is difficult to check measurability of relatons, hence the continuous apRHL is weak for dealing with relation compositions.
However, we have the following two special cases:
\begin{itemize}
	\item The \emph{equality/diagonal} relation \emph{on} any space is a measurable relation.
	\item Any relation between \emph{discrete} spaces is automatically a measurable relation.
\end{itemize}
Hence, the following [comp] rule is an extension of the original [comp] rule in \cite{Barthe:2012:PRR:2103656.2103670}:
\[
	\AxiomC{$
		\begin{array}{l@{}}
			{\Phi} \text{ and } {\Phi'} \text{are measurable relations}\\
			\models c_1 \sim_{(\gamma,\delta)} c_2 \colon \Psi \Rightarrow \Phi \quad
			\models c_2 \sim_{(\gamma',\delta')} c_3 \colon \Psi' \Rightarrow \Phi'
		\end{array}
	$}
	\RightLabel{{[comp]}}
	\UnaryInfC{$
			\models c_1 \sim_{(\gamma\gamma',\min(\delta+\gamma\delta',\delta'+\gamma'\delta))} c_3 \colon \Psi\circ\Psi' \Rightarrow \Phi\circ\Phi' 
	$}
	\DisplayProof
\]
To define the [frame] rule in continuous apRHL, for any relation $\Theta$ on $\interpret{\Gamma}$, we define the following relation $\mathrm{Range}(\Theta)$:
\begin{align*}
	\lefteqn{\mathrm{Range}(\Theta)}\\
	&=\SetBraket{(\nu_1,\nu_2) | \exists{A,B \in \Sigma_{\interpret{\Gamma}}}. ( A\times B\subseteq \Theta \wedge \nu_1(A)=\nu_1(\interpret{\Gamma}) \wedge \nu_2(B)=\nu_2(\interpret{\Gamma}) )}.
\end{align*}
We define the [frame] rule with the construction $\mathrm{Range}({-})$:
\[
	\AxiomC{$
		\begin{array}{l@{}}
			\models c_1 \sim_{(\gamma,\delta)} c_2 \colon \Psi \Rightarrow \Phi \quad
			(\interpret{c_1},\interpret{c_2}) \colon \Theta \to \mathrm{Range}(\Theta)
		\end{array}
	$}
	\RightLabel{{[frame]}}
	\UnaryInfC{$
			\models c_1 \sim_{(\gamma,\delta)} c_2 \colon \Psi\wedge\Theta \Rightarrow \Phi\wedge\Theta
	$}
	\DisplayProof
\]
If $\interpret{\Gamma}$ is countable discrete then the condition $(\nu_1,\nu_2) \in \mathrm{Range}(\Theta)$ is equivalent to $\supp(\nu_1) \times \supp(\nu_2) \subseteq \Theta$, and hence the above [frame] rule is an extension of the original [frame] rule in \cite{Barthe:2012:PRR:2103656.2103670}.

Note that if the $\sigma$-algebra of the space $\interpret{\tau}$ contains all singleton subsets, and $\Theta$ does not restrict any variables in $FV(c_1) \cup FV(c_2)$ then $(\interpret{c_1},\interpret{c_2}) \colon \Theta \to \mathrm{Range}(\Theta)$.
\subsection{Soundness}\label{sec:soundness}
The soundness of the [assn] and [case] are obtained from the composition of arrows in $\BRel(\Meas)$.
The rule [skip] and [seq] are sound because $\overline{\mathcal{G}^{(\gamma,\delta)}}$ is the graded relational lifting of $\mathcal{G}\times\mathcal{G}$ along the forgetful functor $U\colon \BRel(\Meas) \to \Meas^2$. 
The rules [weak] and [op] are sound because $\overline{\mathcal{G}^{(\gamma,\delta)}}$ is monotone with respect to the inclusion order of relations, and preserves opposites of relations.
The soundness of [rand] is proved from Fubini theorem.
The soundness of [cond] is proved by case analyses.
The soundness of [while] is obtained from $\omega\mathbf{CPO}_\bot$-enrichment structure of $\SRel$.
The soundness of [comp] is given by using the measurability of the postconditions.
Finally, the [frame] rule is proved from the strucure of $\mathrm{Range}(\Theta)$.
\subsection{Mechanisms}\label{sec:mechanism}
In this part, we give a generic method to construct the rules for random samplings, and by instantiating the method we show the soundness of the proof rules in prior researches: [Lap] for Laplacian mechanism \cite{DworkMcSherryNissimSmith2006}, [Exp] for Exponential mechanism \cite{McSherry:2007:MDV:1333875.1334185}, [Gauss] for Gaussian mechanism \cite[Theorem 3.22, Theorem A.1]{DworkRothTCS-042}, and [Cauchy] for the mechanism by Cauchy distributions \cite{Nissim:2007:SSS:1250790.1250803}.

Let $f \colon X \times Y \to \mathbb{R}$ be a positive measurable function, and $\nu$ be a measure over $Y$.
We define the following function $f_a \colon \Sigma_Y \to [0,1]$ by 
\[
	f_a(B) = 
	\frac{
		\int_B f(a,-) ~d\nu
	}{
		\int_Y f(a,-) ~d\nu
	}.
\]
We remark that the function $f(a,-)\colon Y \to \mathbb{R}$ is measurable.
If the function is not `almost everywhere zero' and Lebesgue integrable, that is, $0 < \int_Y f(a,-) ~d\nu  < \infty$ then $f_a(-)$ is a \emph{probability measure}.

The following proposition, which is an extension of \cite[Lemma 7]{Barthe:2012:PRR:2103656.2103670}, plays the central role in the construction of sound proof rules for random samplings.
\begin{proposition}\label{prop:mechanism}
Let $f \colon X \times Y \to \mathbb{R}$ be a positive measurable function, and $\nu$ be a measure over $Y$.
For all $a, a' \in X$, $\gamma, \gamma' \geq 1$, $\delta \geq 0$, and $Z\in \Sigma_Y$ (window set),
if the following three conditions hold then for any $B\in \Sigma_Y$, we have $f_a(B) \leq \gamma \gamma' f_{a'}(B) + \delta$.
\begin{enumerate}
	\item $0 < \frac{1}{\gamma'} \int_Y f(a',-)~d\nu \leq \int_Y f(a,-)~d\nu < \infty$
	\item $\forall{b \in Z}. f(a,b) \leq \gamma f(a',b)$, \quad  \rm{(iii)} $f_a(Y \setminus Z) \leq \delta$.
\end{enumerate}
\end{proposition}
\paragraph{Laplacian mechanism \cite{DworkMcSherryNissimSmith2006}.}
We give the function $f \colon \mathbb{R} \times \mathbb{R} \to\mathbb{R}$ by
$f(a,b) = \frac{2}{\sigma}\exp(\frac{-|b - a|}{\sigma})$,
where $\sigma>0$ is the variance of Laplacian mechanism.
We introduce the probabilistic operation $\mathtt{Lap}_\sigma \colon \mathtt{real} \to \mathtt{real}$ with $\interpret{\mathtt{Lap}_\sigma} = f_{(-)}$,
whose measurability is shown from the continuity of the mapping $a \mapsto \int_{\alpha}^{\beta} f(a,x) dx$ ($\alpha, \beta\in\mathbb{R}$).

We show $( f_{(-)}, f_{(-)}) \colon \SetBraket{(a,a') | \left|a - a'\right| < r} \to \overline{\mathcal{G}^{(\exp(\frac{r}{\sigma}),0)}}\mathrm{Eq}_{\mathbb{R}}$
by instantiating Proposition \ref{prop:mechanism} as follows:
If $|a - a'| < r$ then
the following parameters satisfy the conditions (i)--(iii):
$\gamma = \exp(r / \sigma)$, $\gamma' = 1$, $\delta = 0$, 
the function $f$,
the Lebesgue measure $\nu$ over $\mathbb{R}$, and the window $Z = \mathbb{R}$.
This implies $( f_{(-)}, f_{(-)}) \colon \SetBraket{(a,a') | \left|a - a'\right| < r} \to \overline{\mathcal{G}^{(\exp(\frac{r}{\sigma}),0)}}\mathrm{Eq}_{\mathbb{R}}$ since $\SetBraket{(a,a') | \left|a - a'\right| < r}$ and $\mathrm{Eq}_{\mathbb{R}}$ are symmetric.

From the [rand] rule, the following rule is proved:
\[
	\AxiomC{$
			\Gamma \vdash^t e_1 \colon \mathtt{real} \quad \Gamma \vdash^t e_2 \colon \mathtt{real} \quad m_1 \Psi m_2 \Rightarrow |\interpret{e_1} m_1 - \interpret{e_2} m_2| < r
	$}
	\RightLabel{{[Lap]}}
	\UnaryInfC{$
		\models x  \xleftarrow{\$} \mathtt{Lap}_\sigma(e_1) \sim_{(\exp(\frac{r}{\sigma}),0)}  y  \xleftarrow{\$} \mathtt{Lap}_\sigma(e_2) \colon \Psi \Rightarrow x\lrangle{1} = y\lrangle{2}
	$}
	\DisplayProof
\]
\paragraph{Exponential mechanism \cite[Modified]{McSherry:2007:MDV:1333875.1334185}.}
Let $D$ be the discrete Euclidian space $\mathbb{Z}^n$, and $(R,\nu)$ be a (positive) measure space.
Let $q \colon D \times R \to \mathbb{R}$ be a measurable function such that $\sup_{b \in R} |q(a,b)-q(a',b)| \leq  c \cdot ||a - a'||_1$ for some $c > 0$.
Suppose $0 < \int_{R} \exp(\varepsilon q(a,-))~d\nu < \infty$ for any $a \in D$.
We give the function $f \colon D \times R \to\mathbb{R}$ by $f(a,b) = \exp(\varepsilon q(a,b))$, where $\varepsilon>0$ is a constant.
We add the value types $\mathtt{D}$ and $\mathtt{R}$ with $\interpret{\mathtt{D}}^\Gamma = D$ and $\interpret{\mathtt{R}}^\Gamma = R$ to pWHILE, and introduce the probabilistic operation $\mathtt{Exp}_{\lrangle{q,\nu, \varepsilon}} \colon \mathtt{D} \to \mathtt{R}$ with $\interpret{\mathtt{Exp}_{\lrangle{q,\nu, \varepsilon}}} = f_{(-)}$.

We show $( f_{(-)}, f_{(-)}) \colon \SetBraket{(a,a') | \left|| a - a'\right||_1 < r} \to \overline{\mathcal{G}^{(\exp(2\varepsilon r c),0)}}\mathrm{Eq}_{R}$
by instantiating Proposition \ref{prop:mechanism} as follows:
Suppose $||a - a'||_1 < r$.
The following parameters then satisfy the conditions (i)--(iii):
$\gamma = \gamma' = \exp(\varepsilon r c)$, $\delta = 0$,
the function $f$,
the given measure $\nu$,
and the window $Z = R$.

From the [rand] rule, the following rule is proved:
\[
	\AxiomC{$
		\begin{array}{l@{}}
			\Gamma \vdash^t e_1 \colon \mathtt{D} \quad \Gamma \vdash^t e_2 \colon \mathtt{D} \quad
			m_1 \Psi m_2 \Rightarrow  || \interpret{e_1} m_1 - \interpret{e_2} m_2||_1 < r
		\end{array}
	$}
	\RightLabel{{[Exp]}}
	\UnaryInfC{$
		\models x \xleftarrow{\$} \mathtt{Exp}_{\lrangle{q,\nu, \varepsilon}} (e_1) \sim_{(\exp(2\varepsilon r c),0)}  y  \xleftarrow{\$} \mathtt{Exp}_{\lrangle{q,\nu, \varepsilon}} (e_2) \colon \Psi \Rightarrow x\lrangle{1} = y\lrangle{2}
	$}
	\DisplayProof
\]
\paragraph{Gaussian mechanism {\cite[Theorem 3.22, Theorem A.1]{DworkRothTCS-042}}.}
We give the function $f \colon \mathbb{R} \times \mathbb{R} \to\mathbb{R}$ by
$f(a,b) = \frac{1}{\sqrt{2\pi\sigma^2}} \exp(-\frac{(b - a)^2}{2\sigma^2})$,
where $\sigma>0$ is the variance of Gaussian mechanism.
We introduce the probabilistic operation $\mathtt{Gauss}_\sigma \colon \mathtt{real} \to \mathtt{real}$ with $\interpret{\mathtt{Gauss}_\sigma} = f_{(-)}$, whose continuity is easily proved.

We obtain $( f_{(-)}, f_{(-)}) \colon \SetBraket{(a,a') | \left|a - a'\right| < r} \to \overline{\mathcal{G}^{(\gamma,\delta)}}\mathrm{Eq}_{\mathbb{R}}$
by instantiating Proposition \ref{prop:mechanism} as follows:
If $|a - a'| < r$, $1 < \gamma < \exp(1)$, and $\gamma'  = 1$ hold, 
and there is $(3/2) < c$ such that
$2\log(1.25/\delta) \leq c^2$ and $({c r}/{\log\gamma}) \leq \sigma$,
then the parameters $\gamma$, $\gamma'$, and $\delta$,
the function $f$, and
the Lebesgue measure $\nu$ over $\mathbb{R}$
satisfy the conditions (i)--(iii) for the window
$Z = \SetBraket{b | \left|b - (a+a')/2\right| \leq (\sigma^2\log\gamma/r)}$.

From the [rand] rule, we obtain the following rule:
\[
	\AxiomC{$
		\begin{array}{l@{}}
			\exists{c > \frac{3}{2}}.~ ({2\log(\frac{1.25}{\delta}) < c^2}~ \wedge~ {\frac{cr}{\gamma}\leq\sigma}) \quad 1 < \gamma < \exp(1)\\
			\Gamma \vdash^t e_1 \colon \mathtt{real} \quad \Gamma \vdash^t e_2 \colon \mathtt{real} \quad
			m_1 \Psi m_2 \Rightarrow |\interpret{e_1} m_1 - \interpret{e_2} m_2| < r
		\end{array}
	$}
	\RightLabel{{[Gauss]}}
	\UnaryInfC{$
		\models x  \xleftarrow{\$} \mathtt{Gauss}_\sigma(e_1) \sim_{(\gamma,\delta)}  y  \xleftarrow{\$} \mathtt{Gauss}_\sigma(e_2) \colon \Psi \Rightarrow x\lrangle{1} = y\lrangle{2}
	$}
	\DisplayProof
\]
We can relax the above conditions for $c$ to $((1+\sqrt{3})/2) < c$ and $2\log(0.66/\delta) < c^2$ by changing the window $Z$ to
$\SetBraket{b | b \leq (a+a')/2 +  (\sigma^2\log\gamma/r)}$ when $a \leq a'$ and 
$\SetBraket{b | b \geq (a+a')/2 -  (\sigma^2\log\gamma/r)}$ when $a' \leq a$.
\paragraph{Mechanism of Cauchy distributions {\cite{Nissim:2007:SSS:1250790.1250803}}}
We give the function $f \colon \mathbb{R} \times \mathbb{R} \to\mathbb{R}$ by $f (a,b) = \frac{\rho}{\pi ((a-b)^2 + \rho^2)}$.
We introduce the probabilistic operation $\mathtt{Cauchy}_\rho \colon \mathtt{real} \to \mathtt{real}$ with $\interpret{\mathtt{Cauchy}_\rho(e) }^\Gamma m = f_{(-)}$, whose continuity is easily proved.

Let $\gamma = 1+\frac{r^2+r\sqrt{r^2+4\rho^2}}{2\rho^2}$.
We obtain $( f_{(-)}, f_{(-)}) \colon \SetBraket{(a,a') | \left|a - a'\right| < r} \to \overline{\mathcal{G}^{(\gamma,0)}}\mathrm{Eq}_{\mathbb{R}}$
by instantiating Proposition \ref{prop:mechanism} as follows:
If $|a - a'| < r$ then the parameters satisfy the conditions (i)--(iii): $\gamma$, $\gamma' = 1$, $\delta = 0$, the Lebesgue measure $\nu$ over $\mathbb{R}$, and the window $Z = \mathbb{R}$.

From the [rand] rule, we obtain the following rule:
\[
	\AxiomC{$
		\begin{array}{l@{}}
			\Gamma \vdash^t e \colon \mathtt{real} \quad
			m_1 \Psi m_2 \Rightarrow |\interpret{e_1} m_1 - \interpret{e_2} m_2| < r
		\end{array}
	$}
	\RightLabel{{[Cauchy]}}
	\UnaryInfC{$
		\models x  \xleftarrow{\$} \mathtt{Cauchy}_\rho(e_1) \sim_{(\gamma,0)}  y  \xleftarrow{\$} \mathtt{Cauchy}_\rho(e_1)  \colon \Psi \Rightarrow \inverse{(\pi_x \times \pi_y)}(\mathrm{Eq}_{\mathbb{R}})
	$}
	\DisplayProof
\]

\algrenewcommand\algorithmicdo{$\mathtt{do}$}
\algrenewcommand\algorithmicend{}
\algrenewcommand\algorithmicwhile{$\mathtt{while}$}
\algrenewcommand\algorithmicif{$\mathtt{if}$}
\algrenewcommand\algorithmicthen{$\mathtt{then}$}
\algrenewcommand\algorithmicelse{$\mathtt{else}$}
\algrenewcommand\algorithmicfor{$\mathtt{for}$}
\algrenewcommand\algorithmicprocedure{}
\section{An Example: The Above Threshold Algorithm}\label{sec:example}
Barthe, Gaboardi, Gr{\'e}goire, Hsu, and Strub extended the logic apRHL to the logic apRHL+ with new proof rules to describe the \emph{sparse vector technique} (see also \cite[Section 3.6]{DworkRothTCS-042}).
They gave a formal proof of the differential privacy of \emph{above threshold algorithm} in the preprint \cite{2016arXiv160105047B} in arXiv.

In this section, we demonstrate that the above threshold algorithm with \emph{real-valued queries} is proved with \emph{almost the same proof} as in \cite{2016arXiv160105047B}.
The new proof rules of apRHL+ are still sound in the framework of the continuous apRHL.

We consider the following algorithm $\mathtt{AboveT}$:
\begin{algorithm}
	\caption{The Above Threshold Algorithm (\cite{2016arXiv160105047B}, Modified)}\label{alg:above}
\begin{algorithmic}[1]
    \Procedure{$\mathtt{AboveT}$}{$T \colon \mathtt{real}$, $Q \colon \mathtt{queries}$, $d \colon \mathtt{data}$}
    	\State{
			$j \leftarrow 1$;
			$r \leftarrow |Q| + 1$;
			$T \xleftarrow{\$} \mathtt{Lap}_{\varepsilon/2}(t)$;
		}
            \While{$j < |Q|$}
				\State{$S \xleftarrow{\$} \mathtt{Lap}_{\varepsilon/4}(\mathtt{eval}(Q, i, d))$;}
				\If{$T \leq S \wedge r = |Q| + 1$}
				\State{$r \leftarrow j$;
					}
				\EndIf
                \State{$j \leftarrow  j + 1$}
            \EndWhile
    \EndProcedure
\end{algorithmic}
\end{algorithm}

We recall the setting of this algorithm.
This algorithm has two fixed parameters: the threshold $t \colon \mathtt{real}$ and
the set $Q\colon\mathtt{queries}$ of queries where $|Q| \colon \mathtt{int}$ is the number of $Q$.
The input variable is $d\colon\mathtt{int}$, and the output variable is $r \colon \mathtt{int}$.
We prepare the new value types $\mathtt{queries}$ and $\mathtt{data}$ with $\interpret{\mathtt{data}}=\mathbb{R}^N$ and $\mathtt{queries} =\mathtt{int}$ (alias), and the typings $j \colon \mathtt{int}$, $T \colon \mathtt{real}$, and $S \colon \mathtt{real}$.
We assume that an operation $\mathtt{eval} \colon (\mathtt{queries},\mathtt{int}, \mathtt{data}) \to \mathtt{real}$ is given for evaluating $i$-th query in $Q$ for the input $d$.
We require $\interpret{\mathtt{eval}}$ to be \emph{$1$-sensitivity} for the data $d$, that is, $|| d - d' ||_1 \leq 1 \Rightarrow  |\interpret{\mathtt{eval}}(Q, i, d) - \interpret{\mathtt{eval}}(Q, i, d')| \leq 1$.

The differential privacy of $\mathtt{Above}$ is characterised as follows:
\[
\models \mathtt{AboveT} \sim_{\exp(\varepsilon),0}  \mathtt{AboveT} \colon
|| d \langle 1\rangle - d \langle 2 \rangle ||_1 \leq 1 \Rightarrow r \langle 1\rangle = r \langle 2 \rangle.
\]
The following rules in apRHL+ are sound in the framework of continuous apRHL:
\[
	\AxiomC{$
	\forall{i \colon\mathtt{int}}. \models c_1\sim_{\left(\gamma,\delta_i \right)} c_2
	\colon \Psi \Rightarrow (x \langle 1 \rangle  = i \Rightarrow  x \langle 2 \rangle = i)
	\quad
	\sum_{i \colon\mathtt{int}}\interpret{\delta_i} = \delta
	$}
	\RightLabel{{[Forall-Eq]}}
	\UnaryInfC{$
		\models c_1\sim_{\left(\gamma,\delta\right)} c_2
		\colon \Psi \Rightarrow x \langle 1 \rangle = x \langle 2 \rangle
	$}
	\DisplayProof
\]
\[
	\AxiomC{$
		\begin{array}{l@{}}
			\Gamma \vdash^t e_1 \colon \mathtt{real} \quad \Gamma \vdash^t e_2 \colon \mathtt{real} \quad
			m_1 \Psi m_2 \Rightarrow |\interpret{e_1} m_1 + r' - \interpret{e_2} m_2| < r
		\end{array}
	$}
	\RightLabel{{[LapGen]}}
	\UnaryInfC{$
		\models x \xleftarrow{\$} \mathtt{Lap}_\sigma(e_1) \sim_{(\exp(\frac{r}{\sigma}),0)}  y  \xleftarrow{\$} \mathtt{Lap}_\sigma(e_2) \colon \Psi \Rightarrow x \langle 1 \rangle + r' = y \langle 2 \rangle
	$}
	\DisplayProof
\]
\[
	\AxiomC{$
			\Gamma \vdash^t e_1 \colon \mathtt{real} \quad \Gamma \vdash^t e_2 \colon \mathtt{real} \quad
			\quad x \notin FV(e_1) \quad y \notin FV(e_2)
	$}
	\RightLabel{{[LapNull]}}
	\UnaryInfC{$
		\models x \xleftarrow{\$} \mathtt{Lap}_\sigma(e_1) \sim_{(1,0)}  y  \xleftarrow{\$} \mathtt{Lap}_\sigma(e_2) \colon \Psi \Rightarrow x \langle 1 \rangle - y \langle 2 \rangle =  e_1 \langle 1 \rangle - e_2 \langle 2 \rangle
	$}
	\DisplayProof
\]
Hence we extend the contiuous apRHL by adding these rules, and therefore we construct a formal proof almost the same proof as in \cite{2016arXiv160105047B} in the extended continous apRHL.

The soundness of the rule [Forall-Eq] is proved from the following lemma:
\begin{lemma}[{\cite[Proposition 6]{2016arXiv160105047B}, Modified}]\label{lem:forallEq}
If $x \colon \tau$ and the space $\interpret{\tau}$ is countable discrete then
\[
{\bigcap_{i \in \interpret{\tau}}\mathcal{G}^{(\gamma,\delta_i)} (x\langle 1\rangle = i \Rightarrow x\langle 2\rangle = i)}\subseteq{\mathcal{G}^{(\gamma,\sum_{i \in \interpret{\tau}}\delta_i)} (x\langle 1\rangle = x\langle 2\rangle)}.
\]
\end{lemma}
The soundness of the rule [LapGen] is proved from the rules [Lap] and [assn] and the semantically equivalence $\interpret{x \xleftarrow{\$} \mathtt{Lap}_\sigma(e + r'); x \leftarrow x - r'} =  \interpret{x \xleftarrow{\$} \mathtt{Lap}_\sigma(e)}$.
The soundness of [LapNull] is proved by using the [LapGen] and [Frame] rules.
\paragraph{Formal Proof}
We now demonstrate that the $(\varepsilon,0)$-differential privacy of algorithm $\mathtt{AboveT}$ is proved with almost the same proof as in \cite{2016arXiv160105047B}.

From the [Forall-Eq] rule with variable $r$, it suffices to prove for all integer $i$, 
\[
\models \mathtt{AboveT} \sim_{\exp(\varepsilon),0}  \mathtt{AboveT} \colon
|| d \langle 1\rangle - d \langle 2 \rangle ||_1
\leq 1
\Rightarrow (r \langle 1\rangle = i \Rightarrow  r \langle 2 \rangle = i).
\]
We denote by $c_0$ the sub-command consisting of the initialization line 2 of $\mathtt{AboveT}$.
From the rules [assn], [LapGen] rule with $r = r' = 1$, and $\sigma= 2/\varepsilon$, [seq], and [frame] we obtain
\[
\models c_0 \sim_{\exp(\varepsilon/2),0}  c_0 \colon || d \langle 1\rangle - d \langle 2 \rangle ||_1 \leq 1 \Rightarrow || d \langle 1\rangle - d \langle 2 \rangle ||_1 \leq 1 \wedge \Psi.
\]
where
\[
\Psi = {T \langle 1\rangle + 1 = T \langle 2 \rangle}
	\wedge {j \langle 1\rangle = j \langle 2 \rangle} \wedge {j \langle 1\rangle = 1}
	\wedge {r \langle 1\rangle = r \langle 2\rangle} \wedge {r \langle 1\rangle = |Q| +1}.
\]
We denote by $c_1$ and $c_2$ the main loop and the body of the main loop respectively
(i.e. $c_1 = \mathtt{while}~(j < |Q|)~\mathtt{do}~c_2$).
We aim to prove the following judgement by using the [while] rule:
\begin{align*}
\models c_1 \sim_{\exp(\varepsilon/2),0}  c_1 \colon  (|| d \langle 1\rangle - d \langle 2 \rangle ||_1 \leq 1 \wedge \Psi) \Rightarrow (r \langle 1\rangle = i \Rightarrow r \langle 2 \rangle =i).
\end{align*}
To prove this, it suffices to show the following cases for the loop body $c_2$:
\begin{enumerate}
\item \label{enum:loop:before}
If $k < i$ then $\models c_2 \sim_{1,0}  c_2 \colon (\Theta \wedge j\langle 1 \rangle = k) \Rightarrow (\Theta \wedge j\langle 1 \rangle > k)$
\item \label{enum:loop:justnow}
If $k = i$ then $\models c_2 \sim_{\exp(\varepsilon/2),0}  c_2 \colon (\Theta \wedge j\langle 1 \rangle = k) \Rightarrow (\Theta \wedge j\langle 1 \rangle > k)$
\item \label{enum:loop:after}
If $k > i$ then $\models c_2 \sim_{1,0}  c_2 \colon (\Theta \wedge j\langle 1 \rangle = k) \Rightarrow (\Theta \wedge j\langle 1 \rangle > k)$
\end{enumerate}
Here, we provide the following \emph{loop invariant} as follows: 
\begin{align*}
\Theta
=
&(j\langle 1 \rangle < i \Rightarrow ((r \langle 1\rangle = |Q| + 1 \Rightarrow r \langle 2 \rangle =|Q| + 1)\wedge (r \langle 1\rangle = |Q| + 1  \vee r \langle 1\rangle < i)))
\\
&\wedge (j\langle 1 \rangle \geq i \Rightarrow (r \langle 1\rangle = i \Rightarrow r \langle 2 \rangle =i))\\
&\wedge || d \langle 1\rangle - d \langle 2 \rangle ||_1 \leq 1
\wedge T\langle 1 \rangle + 1 = T\langle 2 \rangle
\wedge j\langle 1 \rangle = j\langle 2 \rangle
\end{align*}
The judgement in the case (\ref{enum:loop:before}) is proved from the rules [seq], [assn], [cond], and [frame] and the following fact obtained from the [LapNull] rule:
\begin{align*}
\models
& S \xleftarrow{\$} \mathtt{Lap}_{\varepsilon/4}(\mathtt{eval}(Q, i, d)) \sim_{1,0} S \xleftarrow{\$} \mathtt{Lap}_{\varepsilon/4}(\mathtt{eval}(Q, i, d))\colon\\
& (|| d \langle 1\rangle - d \langle 2 \rangle ||_1 \leq 1)
\wedge (T\langle 1 \rangle + 1 = T\langle 2 \rangle)
\Rightarrow 
((S\langle 1 \rangle < T\langle 1 \rangle) \Rightarrow (S\langle 2 \rangle < T\langle 2 \rangle)).
\end{align*}
The case (\ref{enum:loop:justnow}) is proved from the rules [seq], [assn], [cond], and [frame] and the following fact obtained from the [LapGen] rule:
\begin{align*}
\models
& S \xleftarrow{\$} \mathtt{Lap}_{\varepsilon/4}(\mathtt{eval}(Q, i, d)) \sim_{\exp(\varepsilon/2),0} S \xleftarrow{\$} \mathtt{Lap}_{\varepsilon/4}(\mathtt{eval}(Q, i, d))\colon\\
&  (|| d \langle 1\rangle - d \langle 2 \rangle ||_1 \leq 1
\wedge T\langle 1 \rangle + 1 = T\langle 2 \rangle)
\Rightarrow
(S\langle 1 \rangle + 1 = S\langle 2 \rangle \wedge T\langle 1 \rangle + 1 = T\langle 2 \rangle).
\end{align*}
The case (\ref{enum:loop:after}) is proved in the similar way as (\ref{enum:loop:before}).
\begin{ack}
The author thanks Shin-ya Katsumata for many valuable comments and stimulating discussions, 
Marco Gaboardi for helpful suggestions and the introduction of his preprint \cite{2016arXiv160105047B} in arXiv, Masahito Hasegawa, Naohiko Hoshino, and Takeo Uramoto for advices that contributed to improve the writing of this paper.
\end{ack}
\bibliographystyle{plain}
\bibliography{reference}
\newpage
\appendix
This appendix will be deleted from the final version of this paper.
\section{Appendix}
We show some omitted proofs in this paper.
\subsection{Proofs in Section \ref{sec:subgiry}}
\begin{proposition}
The composition of the category $\SRel = \Meas_{\mathcal{G}}$ is continuous with respect to the ordering $\sqsubseteq$.
\end{proposition}
\begin{proof}
Consider a measurable function $h \colon Y \to \mathcal{G}Z$ and an $\omega$-chain
$\{f_n \colon X \to \mathcal{G}Y\}_n$ with respect to $\sqsubseteq$.
We fix $x \in X$.
Since the $\omega$-chain of measures $f_n(x)$ are bounded, and hence it conveges strongly $(\sup_n f_n)(x)$.
This implies that, from the definition of Lebesgue integral, for any $C \in\Sigma_Z$ and $x \in X$, we obtain
\begin{align*}
	(h^\sharp \circ \sup_n f_n)(x)(C)
	&=  (h^\sharp (\sup_n f_n)(x))(C)\\
	&= \int_Y h({-})(C)~d((\sup_n f_n)(x))\\
	&= \sup_n \int_Y h({-})(C)~d(f_n (x))\\
	&= \sup_n (h^\sharp \circ f_n)(x)(C).
\end{align*}
Consider a measurable function $h' \colon X \to \mathcal{G}Y$ and an $\omega$-chain
$\{f_n \colon Y \to \mathcal{G}Z\}_n$ with respect to $\sqsubseteq$.
From the monotone convergence theorem,
for any $C \in\Sigma_Z$ and $x \in X$, we have
\begin{align*}
	(\sup_n f_n)^\sharp \circ h'(x)(C)
	&=  (h^\sharp (\sup_n f_n)(x))(C)\\
	&= \int_Y \sup_n f_n({-})(C)~d(h'(x))\\
	&= \sup_n \int_Y f_n({-})(C)~d(h'(x))\\
	&= \sup_n(f_n^\sharp \circ h')(x)(C).
\end{align*}
\end{proof}
\begin{lemma}
If $f_1,f_2 \colon X \to \mathcal{G}Y$ satisfy $f_1 \sqsubseteq f_2$ then
$f_1 - f_2$ defined by
\[
(f_1 - f_2)(x)(B) = f_1(x)(B) - f_2(x)(B) \quad (\text{for all }x \in X, B \in \Sigma_Y)
\]
is a measurable function $X \to \mathcal{G}Y$.
\end{lemma}
\begin{proof}
For each $x \in X$,
the finiteness of the measures $f_1(x)$ and $f_2(x)$ imply
the countable additibity of $(f_1 - f_2)(x)$ as follows:
\begin{align*}
(f_1 - f_2)(x)(\sum_n B_n)
&= f_1(x)(\sum_n B_n) - f_2(x)(\sum_n B_n)\\
&= \sum_n f_1(x)(B_n) - \sum_n f_2(x)(B_n)\\
&= \sum_n (f_1(x)(B_n) - f_2(x)(B_n))\\
&= \sum_n (f_1 - f_2)(x)(B_n)
\end{align*}
where $\sum_n B_n$ is the union of a countable disjoint collection $B_0$, $B_1$, ....
Therefore $f_1 - f_2$ is at least \emph{a function} of the form $X \to \mathcal{G}Y$.

The $\sigma$-algebra of $\mathcal{G}Y$ is generated by the following countable collection:
\[
	\SetBraket{\nu \in \mathcal{G}Y| \nu(A) \leq \alpha}
	 \quad (A \in \Sigma_Y, \alpha \in [0,1]\cap\mathbb{Q}).
\]
Since $f_1,f_2 \colon X \to \mathcal{G}Y$,
$A^\alpha_i = \inverse{f_i}(\SetBraket{\nu \in \mathcal{G}Y| \nu(A) \leq \alpha})$ is measurable for all $A \in \Sigma_Y$ and $\alpha \in [0,1]\cap\mathbb{Q}$ ($i = 1,2$).
We then calculate
\begin{align*}
\lefteqn{\inverse{(f_1 - f_2)}(\SetBraket{\nu \in \mathcal{G}Y| \nu(A) \leq \alpha})}\\
&= \SetBraket{x \in X| (f_1 - f_2)(x)(A) \leq \alpha}\\
&= \SetBraket{x \in X| f_1(x)(A) - f_2(x)(A) \leq \alpha}\\
&= \SetBraket{x \in X| f_1(x)(A) - \alpha \leq f_2(x)(A)}\\
&= \bigcap_{\beta \in [0,1]\cap\mathbb{Q}}\SetBraket{x \in X| f_2(x)(A) \leq \beta \implies f_1(x)(A) - \alpha \leq \beta}\\
&= \bigcap_{\beta \in [0,1]\cap\mathbb{Q}}\SetBraket{x \in X| f_2(x)(A) \leq \beta \implies f_1(x)(A)\leq \min(1,\alpha + \beta)}\\
&= \bigcap_{\beta \in [0,1]\cap\mathbb{Q}} ((X \setminus A_2^{\beta}) \cup A_1^{\min(1,\alpha + \beta)})
\end{align*}
Hencer, the function $f_1 - f_2$ is measurable.
\end{proof}
\subsection{Proofs in Section \ref{sec:gradedlifting}}
We recall the definition of the \emph{indicator function} $\chi_A \colon X \to [0,1]$ of a subset $A \subseteq X$:
\[
\chi_A(x) =
\begin{cases}
1, & \text{ if } x \in A\\
0, & \text{ if } x \notin A
\end{cases}
\]
The subset $A$ of $X$ is a measurable \emph{if and only if}
the indicator function $\chi_A$ is a measurable function $\colon X \to [0,1]$.
\begin{lemma}\label{lem:codensity-form}
The following equation holds for any $(\Phi,X,Y)$ in $\BRel(\Meas)$: 
\begin{align*}
	\mathcal{G}^{(\gamma,\delta)} \Phi
	&=\SetBraket{
		(\nu_1,\nu_2) 
		| \forall{(f,g) \colon \Phi \to {\leq} \text { in } \BRel(\Meas)} .
			\int_X f~d\nu_1 \leq \gamma\!\int_Y g~d\nu_2 + \delta
	},
\end{align*}
\end{lemma}
\begin{proof}
We recall
\[
\mathcal{G}^{(\gamma,\delta)} \Phi
	= \SetBraket{
		(\nu_1,\nu_2) \in \mathcal{G}X \times \mathcal{G}Y
		| \begin{array}{l@{}}
			\forall{A \in \Sigma_X, B \in \Sigma_Y}.\\
			\Phi(A)\subseteq B \implies \nu_1(A) \leq \gamma \nu_2(B) + \delta
			\end{array}
	}.
\]
($\supseteq$)
Suppose the pair $(\nu_1,\nu_2)$ satisfies $\int_X f~d\nu_1 \leq \gamma\!\int_Y g~d\nu_2 + \delta$ for all $(f,g) \colon \Phi \to {\leq}$ in $\BRel(\Meas)$.
 
Assume that $A \in \Sigma_X$ and $B \in \Sigma_Y$ satisfy $\Phi(A)\subseteq B$.
The indicator functions $\chi_A \colon X \to [0,1]$, $\chi_B\colon Y \to [0,1]$ are measurable,
and satisfy $\chi_A(x) \leq \chi_B(y)$ for any $(x,y) \in \Phi$
because $(x,y)\in\Phi \wedge x \in A \implies y \in \Phi(A)$.
These imply that $(\chi_A, \chi_B)$ is an arrow $\Phi \to {\leq}$ in $\BRel(\Meas)$.
We then obtain $\int_X \chi_A~d\nu_1 \leq \gamma\!\int_Y \chi_B~d\nu_2 + \delta$, which is
equivalent to $\nu_1(A) \leq \gamma \nu_2(B) + \delta$.

($\subseteq$)
Suppose $(\nu_1,\nu_2) \in \mathcal{G}^{(\gamma,\delta)}\Phi$.
Take an arbitrary arrow $(f,g) \colon \Phi \to {\leq}$ in $\BRel(\Meas)$.
We have $\inverse{f}([\beta,1])\in \Sigma_X$ and $\inverse{g}([\beta,1])\in \Sigma_Y$.
We obtain $\Phi(\inverse{f}([\beta,1])) \subseteq \inverse{g}([\beta,1])$ for any $\beta \in [0,1]$ because $(x,y)\in \Phi \wedge f(x) \geq \beta \implies g(y) \geq \beta$.
By the definiton of Lebesgue integration, we calculate as follows:
\begin{align*}
\lefteqn{\int_X f~d\nu_1}\\
&= \sup\SetBraket{\sum_{k=0}^{n} \alpha_k \nu_1(\inverse{f}([\sum_{l=0}^{k}\alpha_l,1])) | n \in \mathbb{N}, \{\alpha_k\}_{k=1}^{n} \text{ s.t. }\sum_{k=0}^{n} \alpha_k  \leq 1, \forall k.(0 \leq \alpha_k)}\\
&\leq \sup\SetBraket{\sum_{k=0}^{n} \alpha_k (\gamma \nu_2(\inverse{g}([\sum_{l=0}^{k}\alpha_l,1])) + \delta) | n \in \mathbb{N}, \{\alpha_k\}_{k=1}^{n} \text{ s.t. }\sum_{k=0}^{n} \alpha_k  \leq 1, \forall k.(0 \leq \alpha_k)}\\
&\leq \gamma\sup\SetBraket{\sum_{k=0}^{n} \alpha_k\nu_2(\inverse{g}([\sum_{l=0}^{k}\alpha_l,1]))| n \in \mathbb{N}, \{\alpha_k\}_{k=1}^{n} \text{ s.t. }\sum_{k=0}^{n} \alpha_k  \leq 1, \forall k.(0 \leq \alpha_k)} + \delta \\
&=\gamma \int_Y g~d\nu_2 + \delta.
\end{align*}
Here, the first and last equality are given by definition of Lebesgue integration.
The first inequallity is obtained from the assumption $(\nu_1,\nu_2)\in\mathcal{G}^{(\gamma,\delta)}\Phi$.
The second inequallity is obtained from the condition $\sum_{k=0}^{n} \alpha_k  \leq 1$.
\end{proof}

\subsection{Proofs in Section \ref{sec:soundness}}
\begin{lemma}
The rule [rand] is sound.
\end{lemma}
\begin{proof}
	We assume $x_1 \neq x_2$ since the soundness is obvious when $x_1 = x_2$.
	We then obtain $\Gamma = \Gamma, x_1\colon\tau, x_2\colon\tau$ from the precondition of the rule [rand].
	Hence, we may assume $\interpret{\Gamma} = \interpret{\Gamma'} \times \interpret{\tau} \times \interpret{\tau}$.
	It suffices to show
	\begin{align*}
		\lefteqn{(m_1,m_2)\in\Psi}\\
		&\implies (\interpret{\Gamma\vdash x_1 \xleftarrow{\$} d(e^1_1,\ldots,e^1_m) }(m_1), \interpret{\Gamma\vdash x_2 \xleftarrow{\$} d(e^2_1,\ldots,e^2_m)}(m_2) )
		\in \mathcal{G}^{(\gamma, \delta)}(\Phi),
	\end{align*}
	where
	\[
		\Phi
		= (x_1\lrangle{1} = x_2\lrangle{2})
		= \SetBraket{(m_1,m_2) | \pi_{x_1}(m_1) = \pi_{x_2}(m_2)}.
	\]
	Let $(m_1,m_2)\in\Psi$ and $A \in \Sigma_{\interpret{\Gamma}}$.
	We have $\Phi(A) = \interpret{\Gamma'} \times \interpret{\tau} \times A_{x_1}$,
	where $A_{x_1} = \SetBraket{\pi_3(m) | m \in A}$.
	We remark that $A_{x_1}$ is measurable, and therefore so is $\Phi(A)$.
	
	We denote by $\nu_i$ the measure $\interpret{d}(\interpret{\Gamma \vdash^t e^i_1}(m_i),\ldots, \interpret{\Gamma \vdash^t e^i_m}(m_i))$ ($i = 1,2$).
	\begin{align*}
		\interpret{\Gamma\vdash x_1 \xleftarrow{\$} d(e^1_1,\ldots,e^1_m)}(m_1)(A)
		&= \mathcal{G}(\rho_{(x\colon\tau, \Gamma)})
			\circ\mathrm{st}_{\interpret{\tau},\interpret{\Gamma}}
			\circ \lrangle{\nu_1,m_1} (A)\\
		&= \mathrm{st}^{\mathcal{G}}_{\interpret{\tau},\interpret{\Gamma}}(\nu_1,m_1)
			(\inverse{\rho_{(x_1\colon\tau,\Gamma)}}(A))\\
		&= (\nu_1 \otimes \delta_{m_1})(\inverse{\rho_{(x_1\colon\tau,\Gamma)}}(A))\\
		&= \int_{\interpret{\tau} \times \interpret{\Gamma}} \chi_{\inverse{\rho_{(x_1\colon\tau,\Gamma)}}(A)}~d(\nu_1 \otimes \delta_{m_1})\\
		&= \int_{a\in\interpret{\tau}} \left(\int_\interpret{\Gamma} \chi_{\inverse{\rho_{(x_1\colon\tau,\Gamma)}}(A)}(a,-) ~d(\delta_{m_1})\right) ~d\nu_1\\
		&= \int_{\interpret{\tau}} f ~d\nu_1\\
		\vspace{2em}
		\interpret{\Gamma\vdash x_2 \xleftarrow{\$} d(e^2_1,\ldots,e^2_m)}(m_2)(\Phi(A))
		&=\interpret{\Gamma\vdash x_2 \xleftarrow{\$} d(e^2_1,\ldots,e^2_m)}(m_2)(\interpret{\Gamma'} \times \interpret{\tau} \times A_{x_1})\\
		&= (\nu_2 \otimes \delta_{m_2})(\inverse{\rho_{(x_2\colon\tau,\Gamma)}}(\interpret{\Gamma'} \times \interpret{\tau} \times A_{x_1}))\\
		&= (\nu_2 \otimes \delta_{m_2})(A_{x_1})\\
		&= \int_{\interpret{\tau}} g ~d\nu_2,
	\end{align*}
		Where, 
		$f= \chi_{\inverse{(\rho_{(x_1\colon\tau,\Gamma)}(-,m_1))}(A)}$ and $g = \chi_{A_{x_1}}$.
		The pair of these arrows $(f, g)$ forms an arrow $\mathrm{Eq}_{\interpret{\tau}} \to \leq$ in $\BRel(\Meas)$.
		Hence we obtain from Lemma \ref{lem:codensity-form},
	\[
		\interpret{\Gamma\vdash x_1 \xleftarrow{\$} d(e^1_1,\ldots,e^1_m)}(m_1)(A)
		\leq \gamma \interpret{\Gamma\vdash x_2 \xleftarrow{\$} d(e^2_1,\ldots,e^2_m)}(m_1)(A) + \delta.
	\]
	Since $A$ is arbitrary, we conclude 
	\[
		(\interpret{\Gamma \vdash x_1 \xleftarrow{\$} d(e^1_1,\ldots,e^1_m)}(m_1), \interpret{x_2 \xleftarrow{\$} d(e^2_1,\ldots,e^2_m)} (m_2))
		\in \mathcal{G}^{(\gamma, \delta)}(\Phi)
	\]
\end{proof}
\begin{lemma}
The rule [cond] is sound.
\end{lemma}
\begin{proof}
	Let $(m_1,m_2) \in \Psi$.
	We have $\interpret{\Gamma\vdash b}(m_1) = \interpret{\Gamma\vdash b'}(m_2)$ from the preconditions of the rule [cond].
	Since
	\[
		\interpret{\Gamma \vdash \mathtt{if}~b~\mathtt{then}~c_1~\mathtt{else}~c_2}
		=\left[\interpret{\Gamma\vdash c_1},\interpret{\Gamma\vdash c_2} \right]
		\circ \cong_{\interpret{\Gamma}}
		\circ\lrangle{\interpret{\Gamma\vdash b} ,\textrm{id}_{\interpret{\Gamma}}},
	\]
	we have the following two cases:
	\begin{enumerate}
	\item When $\interpret{\Gamma\vdash b}(m_1)= \iota_1(\ast)$,
	we obtain
	\begin{align*}
		\interpret{\Gamma \vdash \mathtt{if}~b~\mathtt{then}~c_1~\mathtt{else}~c_2}(m_1)& = \interpret{\Gamma\vdash c_1}(m_1)\\
		\interpret{\Gamma \vdash  \mathtt{if}~b'~\mathtt{then}~c_1'~\mathtt{else}~c_2'}(m_2) &= \interpret{\Gamma\vdash c_1'}(m_2)
	\end{align*}
	We then obtain
	\begin{equation}\label{eq:rule_cond}
		(\interpret{\Gamma \vdash \mathtt{if}~b~\mathtt{then}~c_1~\mathtt{else}~c_2}(m_1), \interpret{\Gamma \vdash \mathtt{if}~b'~\mathtt{then}~c_1'~\mathtt{else}~c_2'}(m_2))
		\in \overline{\mathcal{G}^{(\gamma,\delta)}}\Phi.
	\end{equation}
	\item When $\interpret{\Gamma\vdash b}(m_1)= \iota_2(\ast)$,
	we obtain (\ref{eq:rule_cond}) similarly.
	\end{enumerate}
	\vspace{-\baselineskip}
\end{proof}
\begin{lemma}
	The rule [while] is sound.
\end{lemma}
\begin{proof}
	We first prove by induction on $n$:
	\begin{align}\label{eq:rule_while_1}
		\nonumber
		\models & [\mathtt{while}~b_1~\mathtt{do}~c_1]_n \sim_{(\prod_{k=0}^{n-1},\gamma_k \sum_{k=0}^{n-1} \delta_k)}  [\mathtt{while}~b_2~\mathtt{do}~c_2]_n \colon
		\\
		 & \quad \Theta\wedge b_1\lrangle{1} \wedge e\lrangle{1} \geq k \Rightarrow \Theta \wedge e\lrangle{1} \geq n+k 
	\end{align}
	\begin{description}
		\item[case: $n=0$]
		We obtain $\models \mathtt{null} \sim_{(1,0)} \mathtt{null} \colon \Theta\wedge b_1\lrangle{1} \wedge e\lrangle{1} \geq k \Rightarrow \emptyset$ since $\interpret{\Gamma\vdash \mathtt{null}}$ is the null measure over $\interpret{\Gamma}$.
		We recall that the following equality:
		\[
			[\mathtt{while}~b_i~\mathtt{do}~c_i]_0 = \mathtt{if}~b_i~\mathtt{then}~\mathtt{null}~\mathtt{else}~\mathtt{skip},
		\]
		We obtain from the above equality, (\ref{eq:rule_while_1}) by applying [skip], [cond], and [weak].
		\item[case: $n=m+1$]
		From the precondition of [while] and the soundness of [case],
		\[
			\models c_1 \sim_{(\gamma_m,\delta_m)} c_2 \colon \Theta \wedge (e\lrangle{1} = k) \implies  (e\lrangle{1} > k)
		\]
		By the induction hypothesis,
		\begin{align*}
		\models & [\mathtt{while}~b_1~\mathtt{do}~c_1]_m \sim_{(\prod_{k=0}^{m-1},\gamma_k \sum_{k=0}^{m-1} \delta_k)}  [\mathtt{while}~b_2~\mathtt{do}~c_2]_m \colon
		\\
		 & \quad \Theta\wedge b_1\lrangle{1} \wedge e\lrangle{1} \geq k \Rightarrow \Theta \wedge e\lrangle{1} \geq m+k 
		\end{align*}
		From the soundness of the rule [seq], we obtain
		\begin{align*}
		\models &c_1 ; [\mathtt{while}~b_1~\mathtt{do}~c_1]_m \sim_{(\prod_{k=0}^{m},\gamma_k \sum_{k=0}^{m} \delta_k)}  c_2 ; [\mathtt{while}~b_2~\mathtt{do}~c_2]_m \colon
		\\
		 & \quad \Theta\wedge b_1\lrangle{1} \wedge e\lrangle{1} \geq k \Rightarrow \Theta \wedge e\lrangle{1} \geq m+1+k 
		\end{align*}
		From the soundness of [weak], [cond], and [skip] we conclude (\ref{eq:rule_while_1}).
	\end{description}
	It is obvious that $\Theta \Rightarrow b_1\lrangle{1}  = b_2\lrangle{2}$ implies
	\begin{equation}\label{eq:rule_while_2}
		\models \mathtt{while}~b_1~\mathtt{do}~c_1 \sim_{(1,0)} \mathtt{while}~b_2~\mathtt{do}~c_2
		\colon \Theta\wedge\neg\wedge b_1\lrangle{1} \Rightarrow \Theta\wedge\neg b_1\lrangle{1}.
	\end{equation}
	From (\ref{eq:rule_while_1}) and (\ref{eq:rule_while_2}), and the soundness of [cond] and [seq], we obtain
	\begin{align*}
		\models & [\mathtt{while}~b_1~\mathtt{do}~c_1]_n;
			\mathtt{while}~b_1~\mathtt{do}~c_1
 \sim_{(\prod_{k=0}^{m} \gamma_k,\sum_{k=0}^{m} \delta_k)}  [\mathtt{while}~b_2~\mathtt{do}~c_2]_n;\mathtt{while}~b_2~\mathtt{do}~c_2 \colon
		\\
		 & \quad \Theta\wedge b_1\lrangle{1} \wedge e\lrangle{1} \geq 0 \Rightarrow \Theta \wedge \neg b_1\lrangle{1}
	\end{align*}
	Since $\SRel = \Meas_{\mathcal{G}}$ is $\omega\mathbf{CPO}_\bot$-enriched, for any command $c$ and expression of the type $\mathtt{bool}$, we obtain
	$
	\interpret{\Gamma\vdash[\mathtt{while}~b~\mathtt{do}~c]_n;\mathtt{while}~b~\mathtt{do}~c} = \interpret{\Gamma\vdash\mathtt{while}~b~\mathtt{do}~c}
	.$
	Hence,
	\begin{align*}
\models & \mathtt{while}~b_1~\mathtt{do}~c_1
 \sim_{(\prod_{k=0}^{m} \gamma_k,\sum_{k=0}^{m} \delta_k)}
		\mathtt{while}~b_2~\mathtt{do}~c_2 \colon
		\\
		 & \quad \Theta\wedge b_1\lrangle{1} \wedge e\lrangle{1} \geq 0 \Rightarrow \Theta \wedge \neg b_1\lrangle{1}
	\end{align*}
\end{proof}
\begin{lemma}
	The rule [frame] is sound.
\end{lemma}
\begin{proof}
	Let $(m_1,m_2) \in \Psi \wedge \Theta$, 
	$\nu_1 = \interpret{\Gamma\vdash c_1}(m_1)$, and 
	$\nu_2 = \interpret{\Gamma\vdash c_2}(m_2)$.
	Since $(\nu_1,\nu_2) \in \mathrm{Range}(\Theta)$,
	there exist $A', B' \in \Sigma_{\interpret{\Gamma}}$ such that
	$A' \times B' \subseteq \Theta$, and $\nu_1(C)=\nu_1(C\wedge A')$ and
	$\nu_2(D)=\nu_2(D\wedge B')$ for all $C, D \in \Sigma_{\interpret{\Gamma}}$.
	Suppose that $A, B \in \Sigma_{\interpret{\Gamma}}$ satisfy
	$(\Phi \wedge \Theta)(A) \subseteq B$.
	Since $A' \times B' \subseteq \Theta$, we have $(\Phi \wedge (A' \times B'))(A) \subseteq B$.
	This implies $\Phi(A \wedge A')\wedge B' \subseteq B$.
	Thus, $\Phi(A \wedge A') \subseteq B + (\interpret{\Gamma} \setminus (B \vee B'))$.
	Therefore 
	\begin{align*}
		\nu_1(A)
		&
		= \nu_1(A \wedge A')
		\leq \gamma \nu_2(B + (M \setminus (B \vee B')) + \delta
		\\
		&
		= \gamma \nu_2((B + (M \setminus (B \vee B')) \wedge B') + \delta
		\leq \gamma \nu_2(B \wedge B') + \delta
		\leq \gamma \nu_2(B) + \delta.
	\end{align*}
	Hence, $(\nu_1,\nu_2) \in \mathcal{G}(\Theta \wedge \Phi)$.
	Similarly, we obtain $(\nu_1,\nu_2) \in \opposite{(\mathcal{G}\opposite{(\Theta \wedge \Phi)})}$.
\end{proof}

\subsection{Proofs in Section \ref{sec:mechanism}}
\begin{proposition}(Proposition \ref{prop:mechanism})
Let $f \colon X \times Y \to \mathbb{R}$ be a positive measurable function, and $\nu$ be a measure over $Y$.
For all $a, a' \in X$, $\gamma, \gamma' \geq 1$, $\delta \geq 0$, and $Z\in \Sigma_Y$ (window set),
if the following three conditions hold then for any $B\in \Sigma_Y$, we have $f_a(B) \leq \gamma \gamma' f_{a'}(B) + \delta$.
\begin{enumerate}
	\item $0 < \frac{1}{\gamma'} \int_Y f(a',-)~d\nu \leq \int_Y f(a,-)~d\nu < \infty$
	\item $\forall{b \in Z}. f(a,b) \leq \gamma f(a',b)$
	\item $f_a(Y \setminus Z) \leq \delta$,
\end{enumerate}
\end{proposition}
\begin{proof}
	From the conditions of this proposition, we obtain for each $B \in \Sigma_Y$,
	\begin{align*}
		f_a(B)
		&=
		f_a(B\cap Z) + f_a(B \setminus Z)
		\\
		&\leq
		\frac{
			\gamma \int_{B\cap Z} f(a',-) ~d\nu
		}{
			\int_Y f(a,-) ~d\nu
		}
		+\delta
		\\
		&\leq
		\frac{
			\gamma \int_{B \cap Z} f(a',-) ~d\nu
		}{
			\frac{1}{\gamma'} \int_Y f(a',-) ~d\nu
		}
		+\delta
		\\
		&
		\leq \gamma \gamma' f_{a'}(B) + \delta.
	\end{align*}
\vspace{-\baselineskip}
\end{proof}

\begin{lemma}[Laplacian Mechanism]
If $|a - a'| < r$ then the following parameters satisfy the conditions (i)--(iii):
$\gamma = \exp(r / \sigma)$, $\gamma' = 1$, $\delta = 0$, 
the function $f \colon \mathbb{R}\times\mathbb{R}\to\mathbb{R}$ defined by $f(a,b) = \frac{2}{\sigma}\exp(\frac{-|b - a|}{\sigma})$,
the Lebesgue measure $\nu$ over $\mathbb{R}$, and the window $Z = \mathbb{R}$.
\end{lemma}
\begin{proof}
The conditions (i) is satisfied,
because the function $f(a,-)$ is the density function of Lapcacian distribution, and hence $\int_{\mathbb{R}} f(a,-) d\nu = \int_{\mathbb{R}} f(a',-) d\nu = 1$.

The condition (iii) is automatically satisfied since $\mathbb{R}\setminus Z = \emptyset$.

We now check that the condition (ii) is satisfied.
The triangle inequality $|b - a'| \leq |a - a'|+|b - a|$ and the assumption $|a - a'| < r$ imply:
\[
\frac{f(a,b)}{f(a',b)}
= \exp\left(\frac{|b - a'|-|b - a|}{\sigma}\right)
 \leq \exp\left(\frac{|a - a'|}{\sigma}\right)
 \leq \exp\left(\frac{r}{\sigma}\right)
\]
This implies $f(a,b) \leq \exp(r / \sigma)f(a',b)$.
\end{proof}

\begin{lemma}[Exponential Mechanism]
Let $D$ be the discrete Euclidian space $\mathbb{Z}^n$, and $(R,\nu)$ be a (positive) measure space.
Let $q \colon D \times R \to \mathbb{R}$ be a measurable function such that $\sup_{b \in R} |q(a,b)-q(a',b)| \leq  c \cdot ||a - a'||_1$ for some $c > 0$.
Suppose $0 < \int_{R} \exp(\varepsilon q(a,-))~d\nu < \infty$ for any $a \in D$.

Suppose $||a - a'||_1 < r$.
The following parameters then satisfy the conditions (i)--(iii):
$\gamma = \gamma' = \exp(\varepsilon r c)$, $\delta = 0$,
the function $f \colon D\times R\to\mathbb{R}$ defined by $f(a,b) = \exp(\varepsilon q(a,b))$ with fixed $\varepsilon>0$,
the given measure $\nu$,
and the window $Z = R$.
\end{lemma}
\begin{proof}
The condition (iii) is obviouslly satisfied.

The condtions (i) and (ii) is obtained from the following calculation:
whenever $||a - a'||_1 < r$, we obtain
\begin{align*}
\frac{f(a,b)}{f(a',b)}
&= \exp\left(\varepsilon q(a,b) - \varepsilon q(a',b) \right)
\leq \exp\left(\varepsilon |q(a,b) - q(a',b)| \right)\\
&\leq \exp\left(\varepsilon c ||a - a'||_1 \right)
\leq \exp\left(\varepsilon c r \right)
\end{align*}
\end{proof}

\begin{lemma}[Gaussian Mechanism: Relaxed Result of {\cite[Theorem A.1]{DworkRothTCS-042}}]
If $|a - a'| < r$, $1 < \gamma < \exp(1)$, and $\gamma'  = 1$ hold, 
and $c = \frac{\sigma\log\gamma}{r}$ satisfies $((1+\sqrt{3})/2) < c$ and $2\log(0.66/\delta) < c^2$,
then the parameters $\gamma$, $\gamma'$, and $\delta$,
the function $f \colon \mathbb{R}\times\mathbb{R}\to\mathbb{R}$ defined by $f(a,b) = \frac{1}{\sqrt{2\pi\sigma^2}} \exp(-\frac{(b - a)^2}{2\sigma^2})$, and
the Lebesgue measure $\nu$ over $\mathbb{R}$
satisfy the conditions (i)--(iii) of Proposition \ref{prop:mechanism}
for the window set $Z$ given by
\[
Z =
\begin{cases} 
\SetBraket{b | b \leq (a+a')/2 +  (\sigma^2\log\gamma /r)},& \text{ if } a \leq a'\\
\SetBraket{b | b \geq (a+a')/2 -  (\sigma^2\log\gamma /r)},& \text{ if } a' \leq a.
\end{cases}
\]
\end{lemma}
\begin{proof}
We assume $a' \leq a$ because in the case $a' > a$, we can prove in the similar way as $a' \leq a$.

The conditions (i) is satisfied, because
for each $a \in \mathbb{R}$ the function $f(a,-)$ is the density function of Gaussian distribution, and hence $\int_{\mathbb{R}} f(a,-) d\nu = \int_{\mathbb{R}} f(a',-) d\nu = 1$.

We prove that the given parameters satisafy the condition (ii) of Proposition \ref{prop:mechanism}.
Suppose $Z = \SetBraket{b | b \leq (a+a')/2 +  (\sigma^2\log\gamma /r)}$.
Take an arbitrary $b \in Z$.
We then calculate as follows:
\begin{align*}
\frac{f(a,b)}{f(a',b)}
&= \exp\left(\frac{(b - a')^2 - (b - a)^2}{2\sigma^2}\right)\\
&= \exp\left(\frac{1}{\sigma^2}(a-a')(b - \frac{a+a'}{2})\right)\\
&\leq \exp\left(\frac{r}{\sigma^2}(b - \frac{a+a'}{2})\right)\\
&\leq \exp\left(\frac{r}{\sigma^2}\frac{\sigma^2\log\gamma}{r}\right) \leq \gamma
\end{align*}
This implies $\forall{b \in Z}.f(a,b) \leq \gamma f(a',b)$.

We prove that given parameters satisfy the condition (iii).
Let $H = \frac{a+a'}{2} + \frac{\sigma^2\log\gamma}{r}$, and let
$H' = \frac{a'-a}{2\sigma} + \frac{\sigma\log\gamma}{r}$.

Since $c > ((1+\sqrt{3})/2)$, we have $c- \frac{1}{2c} -1 > 0$.
From $\log\gamma < 1$, we obtain  $c- \frac{\log\gamma}{2c} -1 > 0$
Since $-r < a' - a $, we obtain $H' > 1$, and hence $\log(H') > 0$.

Since $c^2 > 2\log(0.66/\delta)$, we have $c^2 > 2\log(\frac{1}{\delta}\sqrt{\frac{\exp(1)}{2\pi}})$.
This implies $c^2 - 1 >  2\log(\frac{1}{\delta\sqrt{2\pi}})$.
Since $H' > c - \frac{\log\gamma}{2c} > c -\frac{1}{2c}$,
we then obtain $H'^2 > c^2 - 1 > 2\log(\frac{1}{\delta\sqrt{2\pi}})$.
Therefore, we conclude $\log(H') + H'^2/2 > \log(\frac{1}{\delta\sqrt{2\pi}})$.

We then obtain:
\begin{align*}
\lefteqn{\int_{\mathbb{R}\setminus Z}\frac{1}{\sigma\sqrt{2\pi}}\exp\left(-\frac{(x-a)^2}{2\sigma^2}\right)~d\nu}\\
&= \frac{1}{\sigma\sqrt{2\pi}}\int^{\infty}_{H}\exp\left(-\frac{(x-a)^2}{2\sigma^2}\right)~dx\\
&= \frac{1}{\sqrt{2\pi}}\int^{\infty}_{H'}\exp\left(-\frac{b^2}{2}\right)~db\\
&\leq \frac{1}{\sqrt{2\pi}}\int^{\infty}_{H'} \frac{b}{H'}\exp\left(-\frac{b^2}{2}\right)~db\\
&\leq \frac{1}{\sqrt{2\pi}H'} \exp\left(-\frac{H'^2}{2}\right) \leq \delta.
\end{align*}
This implies $f_a(\mathbb{R}\setminus Z) \leq \delta$.
\end{proof}

\subsection{Proofs in Section \ref{sec:example}}
\begin{lemma}(Lemma \ref{lem:forallEq})
If $x \colon \tau$ and the space $\interpret{\tau}$ is countable discrete then
\[
{\bigcap_{i \in \interpret{\tau}}\mathcal{G}^{(\gamma,\delta_i)} (x\langle 1\rangle = i \Rightarrow x\langle 2\rangle = i)}\subseteq{\mathcal{G}^{(\gamma,\sum_{i \in \interpret{\tau}}\delta_i)} (x\langle 1\rangle = x\langle 2\rangle)}.
\]
\end{lemma}
\begin{proof}
Let $\interpret{\Gamma, x \colon \tau} = \interpret{\tau} \times \interpret{\Gamma}$.
Suppose $(\nu_1,\nu_2) \in \bigcap_{i \in \interpret{\tau}}\mathcal{G}^{(\gamma,\delta_i)} (x\langle 1\rangle = i \Rightarrow x\langle 2\rangle = i)$.
Take an arbitrary $A \in \Sigma_\interpret{\Gamma, x \colon \tau}$.
Since $\interpret{\tau}$ is countable discrete, we decompose $A = \sum_{i \in \interpret{\tau}} (\{i\}\times A_i)$. 
We may assume $A_i \neq \emptyset$ because $\{i\}\times \emptyset = \emptyset$.
Since $ (x\langle 1\rangle = i \Rightarrow x\langle 2\rangle = i)(\{i\}\times A_i) = \{i\}\times \interpret{\Gamma}$,
we obtain $\nu_1(\{i\}\times A_i) \leq \gamma \nu_2(\{i\} \times \interpret{\Gamma}) + \delta_i$ for each $i \in \interpret{\tau}$.
By summing them up, we obtain 
$
\nu_1(A)
\leq
\gamma \nu_2((x\langle 1\rangle = x\langle 2\rangle)(A))+ \sum_{i \in \interpret{\tau}} \delta_i
$.
\end{proof}
\end{document}